\documentclass[journal, 10pt]{IEEEtran}

\usepackage[english]{babel}

\usepackage{graphicx}
\usepackage{caption}
\usepackage{subcaption}
\usepackage{amssymb}
\usepackage{cite}
\usepackage{amsthm}
\usepackage{epstopdf}
\usepackage{color}
\usepackage{soul}
\usepackage{balance}
\usepackage{multirow}
\usepackage{multicol}
\usepackage{booktabs}
\usepackage{amsfonts}
\usepackage{epsfig}
\usepackage{url}
\usepackage{mathrsfs}
\usepackage{physics}
\usepackage{hyperref}
\usepackage{todonotes}
\hypersetup{
    colorlinks=true,
    linkcolor=cyan,
    citecolor=blue,
    pdfborder={0 0 0},
}
\usepackage[english]{babel}
\usepackage[utf8]{inputenc}
\usepackage{textcomp}
\usepackage{algorithm}
\usepackage[noend]{algpseudocode}
\usepackage{dirtytalk}
\usepackage{bbm}
\usepackage{algorithm,algcompatible,amsmath}
\algnewcommand\INPUT{\item[\textbf{Input:}]}%
\algnewcommand\OUTPUT{\item[\textbf{Output:}]}%

\newtheorem{definition}{Definition}
\newtheorem{theorem}{Theorem}
\newtheorem{lemma}{Lemma}

\newtheorem{corollary}{Corollary}

\newtheorem{assumption}{Assumption}

\usepackage{textcomp}
\usepackage{xcolor}
\usepackage{multirow}

\title{Enhancing Gradient Variance and Differential Privacy in Quantum Federated Learning}
\author{Duc-Thien~Phan,
        Minh-Duong~Nguyen,
        Quoc-Viet~Pham,~\IEEEmembership{Senior Member,~IEEE} \\
        and Huilong Pi,~\IEEEmembership{Senior Member,~IEEE} 
        \vspace{-0.25cm}
\thanks{Duc-Thien~Phan is with the College of Computer Science and Electronic Engineering, Hunan University, Hunan 410082, China (e-mail: phanducthien82@hnu.edu.cn).}
\thanks{Minh-Duong Nguyen is with the Department of Intelligent Computing and Data Science, VinUniversity, Hanoi, Vietnam (e-mail: mduongbkhn@gmail.com).}
\thanks{Quoc-Viet Pham is with the School of Computer Science and Statistics, Trinity College Dublin, Dublin 2, D02 PN40, Ireland (e-mail: viet.pham@tcd.ie).}
\thanks{Huilong Pi (corresponding author) is with the College of Computer Science and Electronic Engineering, Hunan University, Hunan 410082, China (e-mail: phl880217@hnu.edu.cn).}
}
\begin{document}
\maketitle

\begin{abstract}
\textcolor{blue}{
Upon integrating Quantum Neural Network (QNN) as the local model, Quantum Federated Learning (QFL) has recently confronted notable challenges. Firstly, exploration is hindered over sharp minima, decreasing learning performance. Secondly, the steady gradient descent results in more stable and predictable model transmissions over wireless channels, making the model more susceptible to attacks from adversarial entities. Additionally, the local QFL model is vulnerable to noise produced by the quantum device's intermediate noise states, since it requires the use of quantum gates and circuits for training. This local noise becomes intertwined with learning parameters during training, impairing model precision and convergence rate. To address these issues, we propose a new QFL technique that incorporates differential privacy and introduces a dedicated noise estimation strategy to quantify and mitigate the impact of intermediate quantum noise. Furthermore, we design an adaptive noise generation scheme to alleviate privacy threats associated with the vanishing gradient variance phenomenon of QNN and enhance robustness against device noise. Experimental results demonstrate that our algorithm effectively balances convergence, reduces communication costs, and mitigates the adverse effects of intermediate quantum noise while maintaining strong privacy protection. Using real-world datasets, we achieved test accuracy of up to 98.47\% for the MNIST dataset and 83.85\% for the CIFAR-10 dataset while maintaining fast execution times.}
\end{abstract}


\section{Introduction} \label{sec:introduction}
The advancement of cloud computing has significantly addressed the challenges of big data processing and enabled the widespread adoption of machine learning (ML) across various domains, including personal fitness tracking~\cite{yadav2023tinyradar}, traffic monitoring~\cite{zhao2024faashark}, and the financial sector~\cite{herman2023quantum}. However, as data volumes continue to increase, ML models are increasingly constrained by computational capacity. Moreover, traditional ML models often pose serious concerns regarding data privacy and security due to the potential leakage of clients' personal information.

Quantum computing, characterized by parallelism and superposition, offers substantial improvements in computational efficiency, positioning Quantum Machine Learning (QML) as a promising research direction in the big data era~\cite{cerezo2022challenges}. Recent works have explored the integration of Quantum Computing (QC) and Federated Learning (FL) to jointly address performance and privacy challenges more effectively than classical ML~\cite{qiao2024transitioning}. This integrated approach, known as Quantum Federated Learning (QFL), remains under-explored with respect to privacy and security preservation.

Existing privacy-preserving techniques for QFL primarily focus on quantum homomorphic encryption~\cite{li2024secure}, quantum differential privacy (DP)~\cite{watkins2023quantum}, and quantum-safe multi-party computation~\cite{sutradhar2021efficient}. The concept of DP involves introducing random noise to model updates, rendering it difficult for adversaries to extract private information~\cite{dwork2016calibrating}. Zhou \textit{et al.}~\cite{zhou2017differential} defined quantum DP and theoretically derived it under three types of quantum interference: amplitude damping, phase damping, and depolarizing. They demonstrated that while quantum noise introduces computational complexity, it simultaneously provides a natural avenue for DP in quantum systems.

Subsequent studies, such as~\cite{du2022quantum}, implemented DP mechanisms based on depolarization noise to protect the outputs of quantum classifiers, including applications to Lasso regression. Senekane~\cite{senekane2017privacy} proposed a privacy-preserving QML framework using DP and validated its effectiveness on a breast cancer dataset. Furthermore, inherent quantum noise resulting from unwanted or imperfect interactions with the physical environment has recently been considered a viable source for achieving DP~\cite{yang2023improved}.

Despite its potential, current DP-based techniques in QML still face significant limitations in the QFL context. Approaches that rely on fixed noise addition~\cite{watkins2023quantum} or uncontrolled inherent noise~\cite{yang2023improved} struggle with the quantum-specific phenomenon of gradient variance vanishing—also known as the \textit{barren plateau} problem~\cite{mcclean2018barren}. This phenomenon becomes more prominent with an increasing number of qubits, leading to two main challenges: (i) the added noise may be insufficient to prevent adversarial attacks due to reduced gradient variance, and (ii) the optimization process may become trapped in sharp local minima, from which it cannot escape~\cite{do2023revisiting,2018-LR-SharpMinia}.

To overcome these limitations, we propose a novel framework called \textbf{Adaptive Differential Privacy Quantum Federated Learning (ADP-QFL)}, designed to enhance both learning performance and privacy preservation in QFL. Specifically, our approach introduces an adaptive DP mechanism that adjusts the noise level in response to the observed gradient variance. Setting the noise level too high degrades learning performance, while setting it too low exposes the model to privacy risks. Our ADP strategy dynamically calibrates the noise based on each client's local gradient and model parameters, ensuring an optimal balance.

In our framework, clients compute local gradients and use them to add a client-specific amount of noise. Additionally, the similarity between local models and the global model is assessed against a predefined threshold to determine whether a client should contribute its update to the global aggregation. The overall system architecture is hybrid: a trusted classical server is used for aggregation, while quantum clients (or quantum simulators) train circuit parameters using a hybrid quantum–classical optimization method. This approach leverages the benefits of both quantum and classical systems.

The main contributions of this paper are summarized as follows:
\begin{itemize}
    \item We propose a novel federated learning algorithm, \textbf{ADP-QFL}, which introduces an adaptive mechanism to calculate client-specific Gaussian noise. The noise is incorporated during model updates, enabling model compression while significantly enhancing privacy guarantees without compromising model performance.
    
    \item We provide a rigorous theoretical analysis of the convergence behavior of ADP-QFL under both convex and non-convex loss functions. The derived bounds highlight the impact of client drift, the number of participating clients, and global iterations on convergence. Moreover, we demonstrate how strategically introduced noise mitigates the \textit{barren plateau} problem when scaling to high-dimensional quantum models.
    
    \item We validate the effectiveness of ADP-QFL through extensive experiments on benchmark datasets. Our results show that ADP-QFL outperforms baseline methods such as FedAvg~\cite{mcmahan2017communication} and FedBN~\cite{li2021fedbn}, achieving $98.47\%$ accuracy on the MNIST dataset and $83.85\%$ on CIFAR-10, while maintaining communication efficiency under the same DP constraints.
\end{itemize}

The remainder of this paper is organized as follows. Section~\ref{sec:related-work} reviews related work. Section~\ref{sec:preliminaries} introduces preliminaries. Section~\ref{sec:proposed-algorithm} presents the proposed ADP-QFL algorithm. Theoretical convergence analysis is discussed in Section~\ref{sec:convergence-analysis}. Experimental results are reported in Section~\ref{sec:experimental}. Finally, Section~\ref{sec:conclusions} concludes the paper. Key mathematical notations are summarized in Table~\ref{tbl:notation} for ease of reference.

\begin{table}[t]
\centering
\caption{Summary of frequently used notations}
\label{tbl:notation}
\begin{tabular}{ll}
\toprule
\textbf{Symbol} & \textbf{Description} \\ \midrule
$\vartheta$       & Gradient variance \\
$\sigma$          & Noise variance \\
$\tau$            & Number of local learning epochs \\
$U$               & Total number of FL clients \\
$K$               & Number of clients selected per training round \\
$\mathcal{D}_i$   & Dataset held by the $i$-th client \\
$\eta_l$          & Client-side learning rate \\
$\eta_g$          & Server-side learning rate \\
$\eta_g^{t}$      & Server learning rate at communication round $t$ \\
$\theta^{(t)}$    & Global model at round $t$ \\
$\theta_*$        & Optimal global model \\
$\theta^{(t,k)}_i$ & Local model of client $i$ at round $t$ and iteration $k$ \\
$T$               & Total number of communication rounds \\
$\mathcal{B}$     & Mini-batch size \\
$L$               & Smoothness constant (Assumption~\ref{ass:L-smooth}) \\
$\mu$             & Strong convexity constant (Assumption~\ref{ass:strongly-convex}) \\
$\sigma_*^2$      & Bounded heterogeneity at optimum (Assumption~\ref{ass:bound-at-optimum}) \\
$P_\text{sharp}$  & Probability of convergence to sharp minimizers \\
$\mathcal{L}$     & Loss function \\ 
$w$ & Global model parameter vector \\
$\mathcal{X}_i$ & Input space of client $i$ \\
$\mathcal{Y}_i$ & Output (label) space of client $i$ \\
$p_i$ & Data distribution on $\mathcal{X}_i \times \mathcal{Y}_i$ for client $i$ \\
\bottomrule
\end{tabular}
\end{table}

\section{Related Works} \label{sec:related-work}
\subsection{Adaptive Differential Privacy}
In FL, two primary definitions of DP are commonly utilized: Sample-Level Differential Privacy (SL-DP) and Client-Level Differential Privacy (CL-DP).

\begin{itemize}
    \item \textbf{SL-DP}: Under SL-DP, the addition or removal of a single sample from a client's local dataset does not significantly affect the output, as each client holds a collection of data samples. In \cite{wang2022safeguarding}, the authors proposed a three-plane framework for cross-silo privacy protection in FL. This framework comprises a Local Differential Privacy (LDP) perturbation algorithm on the client plane and a data reconstruction algorithm on the edge of the client plane to prevent adversarial entities from inferring or directly accessing client data. Furthermore, \cite{li2023multi} introduced a multi-stage adaptive privacy-preserving algorithm that applies DP to Asynchronous FL (AFL) training, preserving both sample-level and client-level privacy.
    
    \item \textbf{CL-DP}: In CL-DP, the inclusion or exclusion of an entire client does not substantially impact the global model's output. Each client contributes a set of local data samples. In \cite{hu2023federated}, a technique was proposed that selects a subset of coordinates from the client's local model updates and applies Gaussian noise to these selected components. Similarly, \cite{shi2023towards} employed gradient perturbation to address the adverse impact of DP and integrated the Sharpness-Aware Minimization (SAM) optimizer to generate locally flat models, which are more stable and robust against weight perturbations. This method yields smaller local update norms and increased resilience to DP noise, thereby improving model performance. Additionally, \cite{yang2021federated} presented a method that allows clients to submit locally differentially private updates at a configurable privacy level, enabling the aggregation server to optimally combine local parameters across heterogeneous privacy settings.
\end{itemize}

In this work, we propose a client-level adaptive noise generation mechanism to mitigate privacy threats caused by the vanishing gradient variance phenomenon in Quantum Neural Networks (QNNs). To address the issue of model stagnation due to the lack of passive learning forces, we introduce a noise estimation technique based on a model approximation strategy. This approach enables dynamic noise adjustment, ensuring that model performance remains above a predefined threshold while preserving privacy.

% \subsection{QFL Frameworks}
Several FL frameworks have been developed to incorporate quantum computing capabilities into distributed learning environments. In \cite{xia2021quantumfed}, the authors introduced a framework that enables multiple quantum nodes to collaboratively train a global Quantum Neural Network (QNN) model using their respective local datasets. In \cite{chehimi2022quantum}, a comprehensive Quantum Federated Learning (QFL) framework was proposed that operates entirely on quantum data, demonstrating superior performance compared to centralized QML architectures.

The study in \cite{huang2022quantum} presented a QFL algorithm aimed at enhancing the communication efficiency of variational quantum algorithms, particularly in scenarios involving decentralized and non-independent and identically distributed (non-IID) quantum data. In the domain of blind quantum computing, \cite{li2021quantum} introduced a private distributed learning framework allowing clients to leverage the computational power of a quantum server without compromising data privacy. This framework incorporates DP techniques to eliminate the need for client trust in the server. However, a practical limitation arises when clients are unable to prepare single qubits for transmission—such a constraint is common in settings involving classical clients or resource-limited quantum devices. Although progress has been made in enhancing privacy within QML, noise introduced during the quantum learning process continues to degrade model accuracy.

While differential privacy has been extensively studied in classical FL systems, relatively few works have focused on adapting DP for QFL. To the best of our knowledge, the only approach that integrates privacy-preserving mechanisms on a quantum platform for comprehensive protection against data leakage and model inversion attacks is presented in \cite{rofougaran2024federated}. Specifically, the authors proposed a hybrid quantum-classical transfer learning strategy that incorporates DP-SGD into FL on Variational Quantum Circuits (VQCs). This integration exploits the expressive capacity and noise sensitivity of VQCs to perform DP training at each local client. Despite its innovation, a notable limitation of this approach is the use of a static noise level, which can lead to increased computational complexity and prolonged convergence times.
{\color{blue}
In contrast, our work introduces an adaptive DP strategy tailored to QFL systems, aimed at mitigating privacy risks associated with gradient vanishing phenomena. By dynamically adjusting the noise level based on model estimations, our method achieves faster convergence, reduces overall computational overhead during training, and effectively avoids barren plateau regions in the QNN landscape.
}
\section{Background and Preliminary} \label{sec:preliminaries}

\subsection{Differential Privacy}

DP is a rigorous mathematical framework that ensures the privacy of individual data records by introducing controlled randomness. Specifically, DP mitigates the risk of inferring whether any particular instance is present in a dataset based on the output of a computation. It provides strong guarantees for preserving data confidentiality in machine learning (ML) algorithms.

\begin{definition}[Differential Privacy]
A randomized algorithm $A$ is said to satisfy $(\epsilon, \delta)$-differential privacy if, for any two adjacent datasets $D$ and $D'$ differing in at most one record, and for any possible subset of outputs $S$, the following condition holds:
\begin{equation} \label{eq:df}
    \Pr[A(D) \in S] \leq e^{\epsilon} \Pr[A(D') \in S] + \delta,
\end{equation}
where $\epsilon$ is the privacy budget and $\delta$ is a small positive constant that quantifies the probability of violating the privacy guarantee. If $\delta = 0$, the guarantee is referred to as \emph{pure} differential privacy; otherwise, it is called \emph{approximate} differential privacy.
\end{definition}

The privacy budget $\epsilon$ is typically chosen to be a small value (e.g., $0.01$), reflecting the desired strength of privacy protection. A smaller $\epsilon$ implies stronger privacy. To achieve $(\epsilon, \delta)$-DP, several noise-adding mechanisms are commonly employed, such as the Laplace mechanism, the Gaussian mechanism, and the exponential mechanism \cite{niu2021adapdp}. Each of these mechanisms introduces noise according to the specific characteristics of the data and the sensitivity of the query function.

In the context of FL, DP can be applied at different granularities. Two commonly used variants are \emph{sample-level DP} and \emph{client-level DP}, which differ based on how adjacent datasets are defined. Sample-level DP provides privacy for individual data points but often necessitates sharing task-specific information that may be vulnerable in privacy-sensitive applications. Consequently, sample-level DP may be insufficient when the training process requires protection of aggregate statistics or task-level information.

In contrast, client-level DP provides stronger guarantees by treating each client’s dataset as a single unit. This approach ensures that participation of any individual client does not significantly influence the output of the learning process. In this work, we specifically focus on client-level DP to protect the entire local dataset of each client during federated model training.

Our proposed framework integrates an adaptive noise estimation technique derived from the model’s estimation process. This method dynamically calibrates the noise to balance privacy protection, convergence speed, and communication efficiency. As a result, our approach enhances robustness while mitigating the trade-offs typically associated with differential privacy in federated settings.

\subsection{Quantum Machine Learning}

\textit{Quantum computing} offers significant computational advantages over classical computing systems by leveraging the principles of quantum mechanics \cite{pan2023experimental}. At its core is the qubit, the fundamental unit of quantum information. Unlike classical bits, which take on a value of either $0$ or $1$, qubits can exist in a superposition of both states simultaneously, enabling a richer computational paradigm.

A qubit can be described as a linear combination of two basis states, denoted $\ket{0}$ and $\ket{1}$, and is mathematically represented as:
\[
\ket{\psi} = \alpha \ket{0} + \beta \ket{1},
\]
where $\alpha$ and $\beta$ are complex probability amplitudes satisfying the normalization condition $|\alpha|^2 + |\beta|^2 = 1$. The computational power of a quantum system increases exponentially with the number of qubits, making qubit scalability a crucial factor. As of the time of writing, practical quantum computers support on the order of tens to hundreds of qubits, with IBM's most advanced system featuring 443 qubits \cite{ibm_quantum}.

\textit{Quantum Neural Networks} (QNNs) are quantum analogs of classical neural networks designed to process quantum data using quantum gates and circuits. While their functional structure is similar to classical neural networks, QNNs operate in a Hilbert space and manipulate quantum states using unitary transformations. The fundamental unit of QNNs is the \textit{quantum perceptron}, which generalizes the classical perceptron concept under quantum mechanics.

One of the significant challenges in training deep QNNs is the vanishing gradient problem—a phenomenon also observed in classical deep learning. This issue arises when gradients become exceedingly small during backpropagation through many layers, hindering effective learning. Although classical neural networks mitigate this problem through the use of non-linear activation functions, quantum models lack such mechanisms due to the unitarity of quantum operations. Thus, QML requires alternative strategies to address this limitation.

The work by McClean \textit{et al.} \cite{mcclean2018barren} identified this phenomenon in QML as \textit{barren plateaus}, where the expected gradient vanishes exponentially with system size. Their study showed that the probability of encountering a barren plateau grows exponentially with the number of qubits, rendering gradient-based optimization infeasible for large-scale quantum systems. Although choosing optimal initial parameters can alleviate this issue in small-scale QNNs, it remains a fundamental obstacle for scalability.

Figure~\ref{fig:variance-qubit} illustrates the exponential decay in gradient variance with increasing qubit count and circuit depth, emphasizing the severity of the barren plateau problem.

\begin{figure}[t]
    \centering
    \begin{subfigure}[b]{0.235\textwidth}
        \centering
        \includegraphics[width=\textwidth]{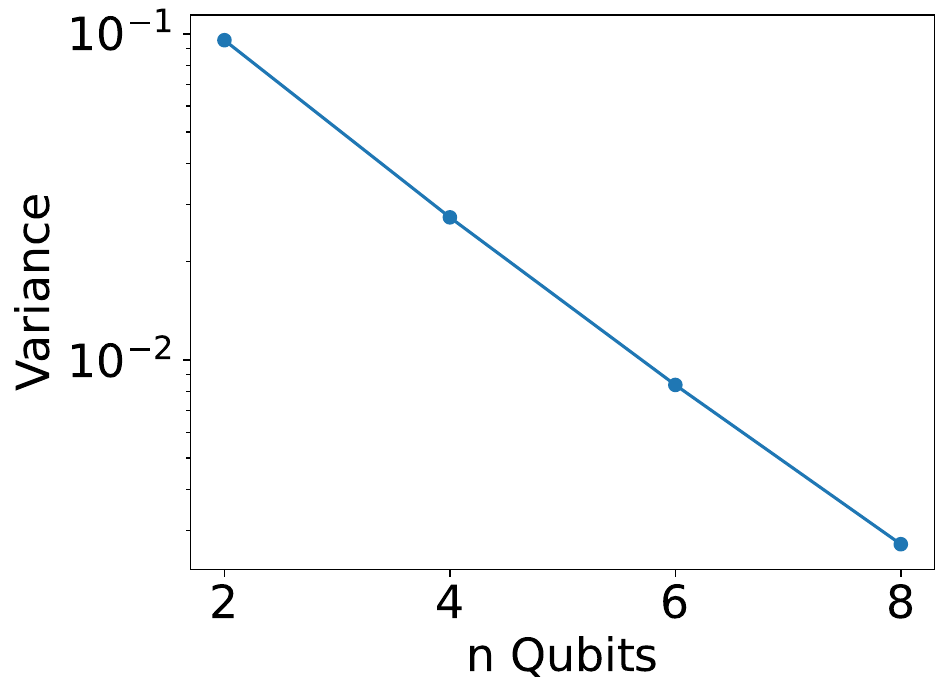}
        \caption{Exponential decay of variance}
        \label{fig:variace-nQubits}
    \end{subfigure}
    \hfill
    \begin{subfigure}[b]{0.235\textwidth}
        \centering
        \includegraphics[width=\textwidth]{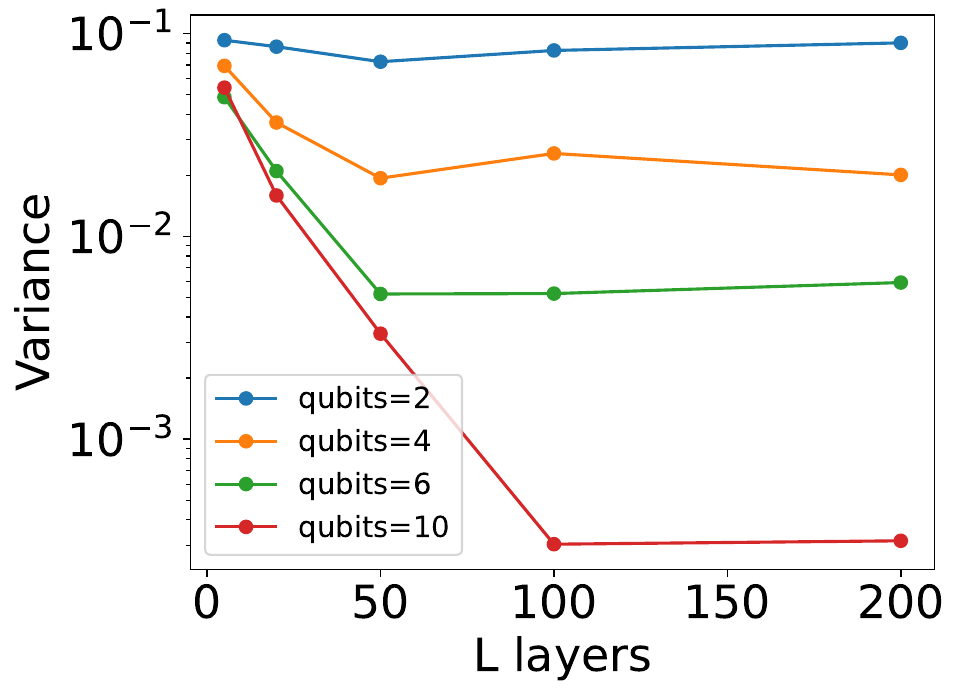}
        \caption{Convergence to 2-design limit}
        \label{fig:variance-layers}
    \end{subfigure}
    \caption{Demonstration of rapid variance decay in QNNs with increasing number of qubits and layers, indicating the presence of barren plateaus.}
    \label{fig:variance-qubit}
\end{figure}

\section{Quantum Federated Learning with Client-level Differential Privacy} \label{sec:proposed-algorithm}

In this section, we propose a novel approach for adaptively regulating DP levels in Quantum Federated Learning (QFL), termed \textbf{ADP-QFL}. The ADP-QFL framework offers a simple yet effective mechanism for enhancing privacy in federated learning systems by introducing customizable noise into the local model updates during transmission.

Specifically, ADP-QFL incorporates a noise injection process into the model update phase, which also enables model compression via adaptive estimation. This dual-purpose mechanism simultaneously strengthens privacy guarantees and reduces communication overhead, all while preserving model performance.

Moreover, we demonstrate that strategically injected noise allows the QFL framework to better adapt to the heterogeneous nature of client data distributions. This adaptive behavior improves both learning efficacy and communication efficiency, making ADP-QFL particularly well-suited for privacy-preserving collaborative learning in quantum-enabled edge environments.

\subsection{Problem Formulation} \label{sec:problem-formulation}

{\color{blue}
Federated Learning (FL) is a decentralized paradigm in which a central server coordinates the training of a shared global model without directly accessing raw data from individual clients~\cite{mcmahan2017communication}. Each client $i \in \{1, \dots, U\}$ maintains a private dataset $\mathcal{D}_i = \{(x_{i,j}, y_{i,j})\}_{j=1}^{N_i}$ and seeks to collaboratively minimize a global loss function defined as:
\begin{align} \label{eq:global-loss}
    &\min_{w \in \mathbb{R}^d} \quad f(w) := \frac{1}{U} \sum_{i=1}^{U} f_i(w), \\
    &\text{where} \quad
    f_i(w) := \mathbb{E}_{(x,y) \sim p_i}[\ell(w; x, y)]. \notag
\end{align}
Here, $w \in \mathbb{R}^d$ denotes the global model parameters, $\ell(w; x, y)$ is the sample-wise loss function, and $p_i$ represents the local data distribution of client $i$. In practice, $p_i \ne p_j$ for $i \ne j$, leading to a non-IID setting that challenges convergence and generalization.

This work considers a Quantum Federated Learning (QFL) scenario, where each client employs quantum neural networks (QNNs) for local model training. The QFL system, illustrated in Figure~\ref{fig:QFL_architecture}, integrates user-level differential privacy (UDP) and a model estimation mechanism to ensure privacy and communication efficiency. During each communication round, a subset of clients $\mathcal{S}_t \subseteq \mathcal{C}$ is selected to participate. Each client performs local training involving:
\begin{enumerate}
    \item \textit{Quantum-based function estimation},
    \item \textit{Adaptive noise injection based on training stage}, and
    \item \textit{Local model weight computation using meta-learning}.
\end{enumerate}

To select informative clients while maintaining fairness, we employ stratified sampling guided by the Fisher Information Matrix (FIM). Each client $u \in \mathcal{C}$ estimates its FIM as:
\[
\mathcal{I}_u(\theta) = \mathbb{E}_{\mathcal{D}_u} \left[
\nabla_\theta \log p(\mathcal{D}_u | \theta) \nabla_\theta \log p(\mathcal{D}_u | \theta)^\top
\right],
\]
with selection probability:
\[
P(u) = \frac{\operatorname{Tr}(\mathcal{I}_u(\theta))}{\sum_{v \in \mathcal{C}} \operatorname{Tr}(\mathcal{I}_v(\theta))}.
\]

Local training proceeds by computing client-specific updates. Each client computes the first-order gradient:
\[
\widetilde{\theta}_u = \theta_u - \eta \nabla_{\theta_u} \mathcal{L}(f_{\theta_u}; \mathcal{D}^s_u),
\]
followed by a second-order meta update:
\begin{align}
\phi_u = \theta_u - \beta \left[\mathbb{I} - \eta \nabla^2_{\widetilde{\theta}_u} \mathcal{L}(f_{\widetilde{\theta}_u}; \mathcal{D}^s_u)\right]
\left[\nabla_{\widetilde{\theta}_u} \mathcal{L}_u(f_{\widetilde{\theta}_u}; \mathcal{D}^q_u) + G_t\right],
\end{align}
where $\mathcal{D}^s_u$ and $\mathcal{D}^q_u$ are the support and query sets, respectively, and $G_t$ is the injected Gaussian noise for differential privacy.

To adapt privacy protection throughout training, we propose an adaptive noise scheme:
\[
\sigma_t^2 = \sigma_0^2 \cdot \frac{1}{1 + \alpha t},
\]
where $\sigma_0^2$ is the initial variance and $\alpha > 0$ is the decay rate. This ensures larger noise during early rounds—when gradients are large and less sensitive—and reduced noise in later stages for precise fine-tuning.

Only clients whose local model estimates satisfy:
\[
\left\| \hat{\phi}_u - \phi_u \right\| \leq \frac{b}{\lambda}
\]
are selected for server aggregation, where $b$ is the estimation error and $\lambda$ is a regularization parameter.

\begin{figure}[t]
\centering
\includegraphics[width=0.95\linewidth]{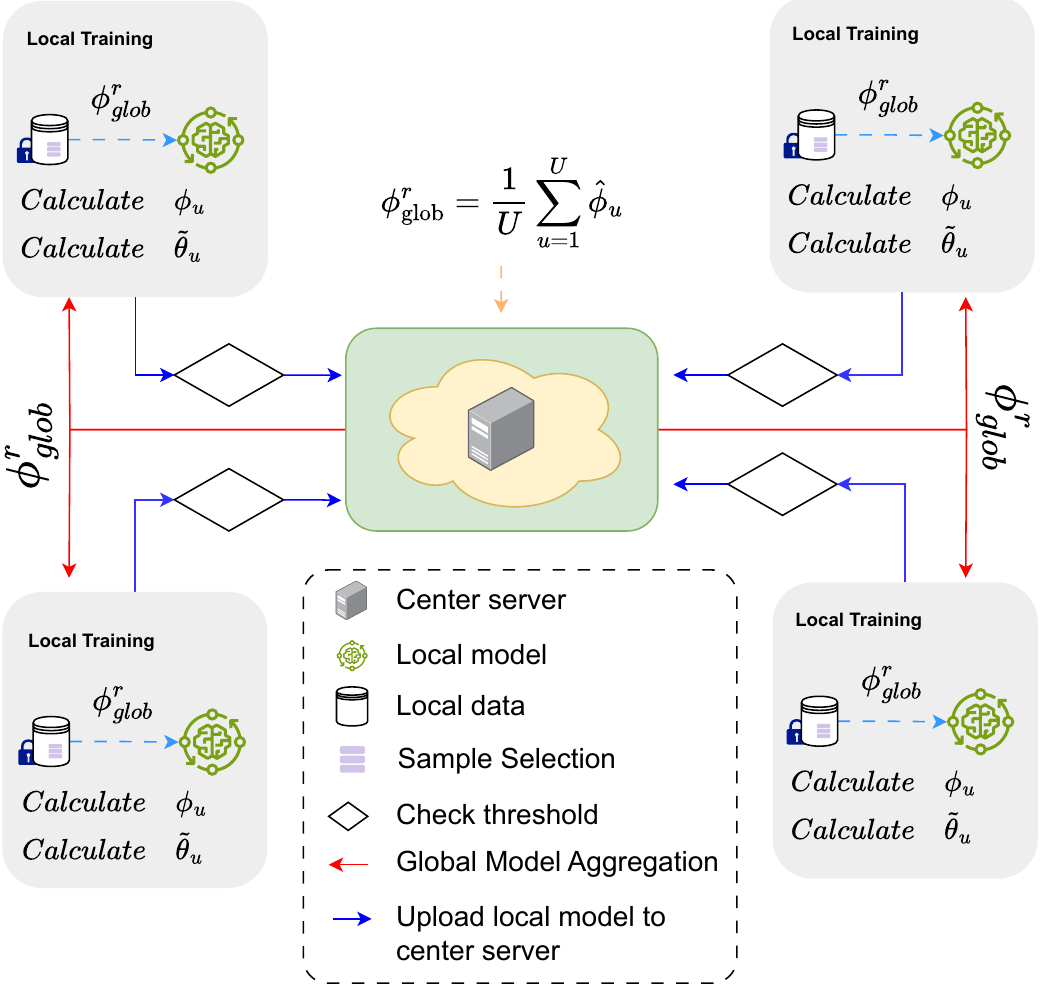}
\caption{Overview of the proposed ADP-QFL framework. In each round, selected clients train locally using quantum models, apply adaptive noise, and transmit filtered model updates to the server.}
\label{fig:QFL_architecture}
\end{figure}
}

\subsection{Architecture}

The proposed architecture comprises three main phases implemented across the server and participating clients, as depicted in Figure~\ref{fig:QFL_architecture} and detailed in Algorithm~\ref{alg:PerTDP}. Our privacy-preserving QFL framework integrates User-level Differential Privacy (UDP) to ensure that each client's private data remains protected throughout the training process.

\subsubsection{Initialization on the Server Side}

At the beginning of the training process, the server initializes the global model parameters along with the key hyperparameters required for privacy control. These include the noise variance $\sigma^2$, the estimation error bound $b$, and the regularization coefficient $\lambda$. The parameters $\sigma$, $b$, and $\lambda$ are used to balance the trade-off between model utility and privacy protection. Further details on the role and impact of these parameters are discussed in Subsection~\ref{subsec:d}.

\subsubsection{Local Training}

For each communication round, a subset of clients is selected. Each participating client initializes its local model using the received global model broadcasted by the server. Subsequently, each client selects mini-batches from its local data pool for training, as outlined in \textit{Algorithm~\ref{alg:PerTDP}, line~\ref{lst:line:query_set}}.

The QFL model on each client computes gradients based on the local training dataset. Client-specific parameters are derived via a meta-update procedure, as described in \textit{Algorithm~\ref{alg:PerTDP}, lines~\ref{lst:line:calc_gradient}-\ref{lst:line:calc_meta_para}}.

To preserve differential privacy, the noise generation process is integrated with a utility-aware sampling mechanism. This step determines optimal parameter values for both noise injection and sampling, aiming to maximize utility while protecting client data. The framework performs multiple iterations using the selected mechanism and computed parameters to achieve effective privacy-preserving training.

\textcolor{blue}{Following noise integration, sparse estimation techniques are employed to select and compress local models, as indicated in \textit{Algorithm~\ref{alg:PerTDP}, line~\ref{lst:line:calc_estimate}}. These techniques enable adaptive learning of new tasks or classes without degrading the performance on previously acquired knowledge. The sparsely extracted model information is then used in the subsequent model selection phase to further enhance efficiency and communication effectiveness.}

\subsubsection{Model Aggregation at the Server}

After local training and noise perturbation on the selected clients, each client uploads its perturbed local model $\hat{\phi}_{u}^t$ to the central server. In this work, we consider a synchronous federated learning framework, wherein the server waits until all selected clients have uploaded their local models before proceeding to the next round.

The server then aggregates the received local models $\{\hat{\phi}_{1}^t, \hat{\phi}_{2}^t, \dots, \hat{\phi}_{K}^t\}$ to obtain the updated global model $\phi_{\text{glob}}^t$. Here, $\hat{\phi}_{i}^t$ denotes the locally updated model from the $i$-th client at communication round $t$. The global model is computed as the average of all received local models. Subsequently, the server broadcasts the updated global model $\phi_{\text{glob}}^t$ to the selected clients for use in the next training iteration. This synchronous update mechanism significantly reduces the computational and storage overhead associated with model aggregation, while ensuring privacy and model consistency.

\begin{algorithm}[t]
\caption{QFL with Task-level Differential Privacy}
\label{alg:PerTDP}
\begin{algorithmic}[1]
\REQUIRE Server-Side Aggregation \label{req:server}
\STATE Initialize global model $\phi^0_{\textrm{glob}}$.
\STATE Initialize noise variance $\sigma^2 = \frac{8T(2L + b)^2 \log(1/\delta)}{K^2 \epsilon^2}$, estimation error $b$, and regularization parameter $\lambda$.
\FOR{each communication round $r = 1,2,\dots$}
    \STATE Receive client models $\{\hat{\phi}_u^r\}_{u=1}^U$ from selected clients.
    \STATE Aggregate to update global model:
    \[
        \phi^r_\textrm{glob} = \frac{1}{U} \sum_{u=1}^{U} \hat{\phi}_u^r.
    \]
    \STATE Broadcast $\phi^r_\textrm{glob}$ and noise variance $\sigma^2$ to selected clients.
    \STATE Proceed with local training on clients.
\ENDFOR
% \\[0.5em]
\REQUIRE Client-Side Local Training
\STATE Set initial model $\theta_u = \phi^r_{\textrm{glob}}$.
\STATE Sample support set $\mathcal{D}^{s}_{u} = (\textbf{X}^{s}_{u}, \textbf{Y}^{s}_{u})$ and query set $\mathcal{D}^{q}_{u} = (\textbf{X}^{q}_{u}, \textbf{Y}^{q}_{u})$ from local dataset $\mathcal{D}_u$. \label{lst:line:query_set}
\FOR{$k = 1$ to $\tau$}
    \STATE Compute local gradient using support set:
    \[
        \nabla_{\theta_u}\mathcal{L}(f_{\theta_u}; \mathcal{D}^{s}_{u}).
    \] \label{lst:line:calc_gradient}
    \STATE Update intermediate parameter:
    \[
        \widetilde{\theta}_u = \theta_u - \eta \nabla_{\theta_u} \mathcal{L}(f_{\theta_u}; \mathcal{D}^s_u).
    \]
    \STATE Compute meta-based update:
    \begin{align}
        \phi_u = \theta_u - \beta & \left[\mathbb{I} - \eta \nabla^2_{\widetilde{\theta}_u} \mathcal{L}(f_{\widetilde{\theta}_u}; \mathcal{D}^s_u)\right] \notag \\
        & \times \left[\nabla_{\widetilde{\theta}_u} \mathcal{L}_u(f_{\widetilde{\theta}_u}; \mathcal{D}^q_u) + G_t\right].
    \end{align} \label{lst:line:calc_meta_para}
    \STATE Apply sparsification and estimate:
    \[
        \hat{\phi}_u \quad \text{such that} \quad \|\hat{\phi}_u - \phi_u\| \leq \frac{b}{\lambda}.
    \] \label{lst:line:calc_estimate}
\ENDFOR
\STATE Upload the perturbed local model $\hat{\phi}_u$ to the server for aggregation.
\end{algorithmic}
\end{algorithm}

\subsection{Client-Level Noise Calculation}
\label{subsec:d}

In a differentially private federated learning (DP-FL) framework, Gaussian noise is added to the model parameters at the client side to preserve privacy. Specifically, during each global communication round $t$, the $u$-th client ($u \in \{1, \ldots, K\}$) adds Gaussian noise drawn from $\mathcal{N}(0, \sigma_u^2)$ to its locally trained parameters before transmitting the perturbed model to the server.

It is assumed that the noises introduced by different clients across communication rounds are independent and identically distributed (IID). Moreover, each noise component is considered to be independent of the model parameters prior to its application, thereby ensuring the integrity of the privacy mechanism.

\begin{lemma}[Utility Guarantee for Differential Privacy \cite{2022-MeL-TaskLevelDP}]
\label{lem:added-noise}
To ensure $(\epsilon, \delta)$-user-level differential privacy (UDP) after $T$ global communication rounds, the minimum variance $\sigma_u^2$ of the Gaussian noise added to the local parameters at each selected client is given by:
\begin{align}
    \sigma^2 = \frac{8T(2L + b)^2 \log(1/\delta)}{K^2 \epsilon^2},
\end{align}
where $T$ is the total number of global iterations, $L$ is the Lipschitz constant of the loss function, $b$ denotes the estimation error, $K$ is the number of participating clients per round, and $\epsilon$ is the privacy budget.
\end{lemma}

\noindent Lemma~\ref{lem:added-noise} indicates that the noise variance $\sigma^2$ must be carefully calibrated to guarantee $(\epsilon, \delta)$-differential privacy during the federated learning process.

In practice, quantum neural networks (QNNs) trained on parameterized quantum circuits may experience gradient variance loss, which can degrade model performance. Interestingly, the addition of noise not only enhances privacy but can also introduce beneficial variance into the training dynamics, thereby improving convergence and generalization. This phenomenon aligns with the principles of adaptive noise injection in differentially private learning systems, as discussed in \cite{watkins2023quantum}.

Hence, the noise calibration strategy described in Lemma~\ref{lem:added-noise} is particularly advantageous when applying QNNs in federated learning scenarios involving a large number of quantum devices. The capability to regulate noise during training contributes both to preserving privacy and to stabilizing optimization under the inherent uncertainties of quantum computation.

{\color{blue}
In this study, we propose an adaptive noise mechanism that dynamically adjusts the injected noise based on the current training stage. Specifically, the noise variance at iteration $t$ is defined as:
\[
\sigma_t^2 = \sigma_0^2 \cdot \frac{1}{1 + \alpha t},
\]
where $\sigma_0^2$ denotes the initial noise variance, $\alpha > 0$ is a decay rate hyperparameter controlling the speed at which noise decreases, and $t$ represents the current iteration index. 

This formulation captures the intuition that in the early stages of training, when the model parameters are far from optimal and gradient magnitudes are large, adding higher noise levels has minimal impact on overall model utility. As training progresses and the model approaches convergence, gradients become smaller and more sensitive to perturbation; thus, the noise variance is reduced accordingly. 

By adapting the noise level dynamically, this strategy provides stronger privacy protection in the early rounds (where model updates are more diverse and harder to trace) and better model accuracy in the later rounds (where smaller updates are critical for fine-tuning). This approach ensures a better privacy-utility trade-off throughout the entire training process.
}

\subsection{Gradient Estimation for Transmission Model Compression}
\label{subsec:grad-estimation}

Effective management of the privacy budget is essential in the implementation of DP, as an excessive expenditure of the budget can significantly compromise the utility of the underlying private dataset. In this work, we extend the standard DP framework by introducing a mechanism to optimally allocate privacy budgets across data consumers. Rather than transmitting the complete local model with added noise to the server for aggregation, we propose a gradient estimation approach to reduce the communication burden by compressing the transmitted model.

The following definition formalizes the phenomenon of *privacy budget explosion*, which occurs as the number of participating clients and training rounds increases.

\begin{definition}[Privacy Budget Explosion]
\label{def:privacy-budget}
The global privacy budget $\epsilon_{\mathrm{glob}}$ increases significantly as the number of clients and communication rounds grows, and is computed as:
\begin{equation}
    \epsilon_{\mathrm{glob}} = \sum_{t=1}^{\tau} \sum_{u=1}^{K} \epsilon_u^t,
\end{equation}
where $\epsilon_u^t$ denotes the privacy budget consumed by client $u$ during communication round $t$, and $\tau$ is the total number of global training rounds.
\end{definition}

\noindent As indicated in Definition~\ref{def:privacy-budget}, directly adding noise to the local model and transmitting it to the server leads to an accumulation of privacy cost over time, potentially exhausting the available privacy budget. To mitigate this issue, we propose a *model estimation* strategy that filters and compresses the local models before transmission to the server.

Inspired by the well-known Random Sample Consensus (RANSAC) method \cite{fischler1981random}, our proposed Adaptive Differential Privacy for Quantum Federated Learning (ADP-QFL) algorithm employs two main components for model estimation: (i) data sampling and (ii) model selection. Initially, subsets of the training data are randomly sampled to construct candidate hypotheses. Each local model is then trained and compared to the global model. The difference between local and global model parameters is computed and compared to a predefined threshold.

This threshold, acting as a filter, serves to eliminate local models that significantly deviate from the global objective, thereby reducing both communication bandwidth and transmission time. The model estimation process is repeated iteratively, and the best-performing local model (based on the threshold criterion) is selected for transmission to the server.

In the ADP-QFL algorithm (see \textit{Algorithm~\ref{alg:PerTDP}}, line~\ref{lst:line:calc_estimate}), the model estimation condition is formalized as:
\begin{align}
    \left\| \hat{\phi}_u - \phi_u \right\| \leq \frac{b}{\lambda},
\end{align}
where $b$ represents the estimation error, $\lambda$ is a regularization parameter, $\phi_u$ denotes the original local model parameters, and $\hat{\phi}_u$ is the estimated model after local training. The norm $\left\| \hat{\phi}_u - \phi_u \right\|$ quantifies the deviation between pre- and post-training model weights.

Only those local models whose deviation remains within the specified threshold are selected for aggregation. This mechanism not only reduces communication overhead but also improves global model quality by filtering out noisy or unstable local updates, particularly in large-scale federated learning systems with many clients and local updates.

\section{Theoretical Analysis}
\label{sec:convergence-analysis}

In this section, we present a theoretical analysis of the proposed algorithm, focusing on gradient variance reduction in quantum federated learning (QFL). Furthermore, we establish convergence guarantees under both convex and non-convex optimization settings.

\subsection{Assumptions and Definitions}
\label{subsec:assumptions}

To facilitate the convergence analysis, we adopt the following standard assumptions commonly used in federated learning literature.

\begin{assumption}[$L$-Smoothness]
\label{ass:L-smooth}
Each local objective function $F_u(w)$ is differentiable and $L$-smooth for all clients $u \in \{1, 2, \ldots, U\}$. That is, for all $w, w' \in \mathbb{R}^d$,
\begin{equation}
    \left\| \nabla F_u(w) - \nabla F_u(w') \right\| \leq L \left\| w - w' \right\|.
\end{equation}
\end{assumption}

\begin{assumption}[$\mu$-Strong Convexity]
\label{ass:strongly-convex}
Each local objective function $F_u(w)$ is differentiable and $\mu$-strongly convex for all clients $u \in \{1, 2, \ldots, U\}$. That is, for all $w, w' \in \mathbb{R}^d$,
\begin{equation}
    \left( \nabla F_u(w) - \nabla F_u(w') \right)^\top (w - w') \geq \mu \left\| w - w' \right\|^2.
\end{equation}
\end{assumption}

\begin{assumption}[Bounded Gradient Dissimilarity at Optimum]
\label{ass:bound-at-optimum}
The heterogeneity among client datasets is bounded at the global optimum $w^*$. Specifically,
\begin{equation}
    \frac{1}{U} \sum_{u=1}^{U} \left\| \nabla F_u(w^*) \right\|^2 \leq \sigma_*^2,
\end{equation}
where $\sigma_*^2$ is a constant that quantifies the level of dissimilarity across client gradients at the optimal point.
\end{assumption}

\subsection{Gradient Variance Reduction in Quantum Federated Learning}
\label{subsec:gradient-variance}

We begin by formalizing the decomposition of stochastic gradients used in our quantum federated learning (QFL) framework.

\begin{definition}[Stochastic Gradient Decomposition]
\label{def:single-gradient}
Let $g_i^t$ denote the stochastic gradient at iteration $t$ computed using a reference data point $x_i$. This gradient can be decomposed into two components:
\begin{equation}
    g_i^t = \bar{g}^t + \Delta g_i^t,
\end{equation}
where $\bar{g}^t$ is the deterministic (general) gradient representing the expected behavior across all data perturbations, and $\Delta g_i^t$ is the variance component capturing the random perturbation introduced by using individual data points.
\end{definition}

Definition~\ref{def:single-gradient} leads to the following theorem, which characterizes the variance reduction effect when using $n$-qubit quantum gradient estimation. This result is grounded on the unbiased gradient assumption discussed in \cite[Assumption 2]{2020-FL-FedNova}.

\begin{theorem}[Unbiased $n$-Qubit Gradient Estimator]
\label{theorem:batch-unbiased-qubit-grad}
Let $\bar{g}^t$ be the expected gradient defined in Definition~\ref{def:single-gradient}, and let $g^t_{\text{$n$-Qubit}}$ denote the batch gradient estimated using an $n$-qubit quantum system. Assuming that each point-wise gradient has bounded variance $\sigma^2$, as in \cite{2020-FL-FedNova}, the stochastic gradient obtained from a mini-batch $\mathcal{B}$ satisfies:
\begin{equation}
    \mathbb{E}_{(x_i, y_i) \sim \mathcal{B}} \left[ \left\| \bar{g}^t - g^t_{\text{$n$-Qubit}} \right\|^2 \right] \leq \frac{3\sigma^2}{2^{2n} - 1}.
\end{equation}
This inequality shows that the variance of the $n$-qubit gradient estimator decreases exponentially with respect to the number of qubits $n$.
\end{theorem}

\begin{proof}
Please refer to Appendix~\ref{Proof:Theorem2} for a complete proof.
\end{proof}

\subsection{Convergence Analysis}
\subsubsection{Convergence under the Convex Setting}

In this subsection, we present the convergence theorem of the proposed algorithm in a convex setting. We also provide a detailed analysis of the contributing components that influence the overall convergence rate.

\begin{theorem}[Convergence Rate under Convexity]
\label{theorem:loss-decrease}
Suppose that Assumptions~\ref{ass:L-smooth} and~\ref{ass:strongly-convex} hold, and that all clients participate fully using full-batch gradients. Let the learning rate $\eta$ satisfy $\frac{1}{2\sqrt{6}\tau^2L} < \eta < \frac{1}{6\tau L}$. Then, the sequence $\{w^r\}$ generated by the ADP-QFL algorithm satisfies the following inequality:
\begin{align}
    & \mathcal{L}(\bar{\theta}^{(R)}) - \mathcal{L}(\theta^*) \notag\\
    & \leq \mathcal{O}\left( \frac{\| \theta^{(0)} - \theta^* \|^2}{\sum_{t=0}^{T-1} \eta_l \tau} \right) 
    + \mathcal{O}\left( \eta_l^2 \tau (\tau - 1) L \sigma_*^2 \right) \notag \\
    & \quad + \mathcal{O}\left( \eta_l \tau \sigma_*^2 \right)
    + \mathcal{O}\left( (2\tau^2 + 3\tau + 1)\eta_l^2 \cdot \frac{3\sigma^2}{2^{2n} - 1} \right).
\end{align}
\end{theorem}

\begin{proof}
The proof is provided in Appendix~\ref{Proof:Theorem3}.
\end{proof}

The convergence bound in Theorem~\ref{theorem:loss-decrease} consists of four key terms, each of which corresponds to a different source of error in the training process:

\begin{itemize}
    \item \textbf{Initialization Error:} 
    \[
    \mathcal{O}\left( \frac{\| \theta^{(0)} - \theta^* \|^2}{\sum_{t=0}^{T-1} \eta_l \tau} \right)
    \]
    This term represents the effect of initialization and is influenced by the total number of communication rounds $R$ and the number of local updates $\tau$. It is invariant across all FL algorithms and does not depend on the specific settings of ADP-QFL.

    \item \textbf{Client Drift Error:} 
    \[
    \mathcal{O}\left( \eta_l^2 \tau (\tau - 1) L \sigma_*^2 \right)
    \]
    This term quantifies the error caused by local model drift due to data heterogeneity across clients. It is influenced by the number of local epochs $\tau$, the smoothness constant $L$, the learning rate $\eta_l$, and the heterogeneity measure $\sigma_*^2$.

    \item \textbf{Noise at Optimum:} 
    \[
    \mathcal{O}\left( \eta_l \tau \sigma_*^2 \right)
    \]
    This term captures the stochastic variance of gradients at the global optimum. It reflects the intrinsic noise in the data distribution and is independent of algorithm-specific parameters.

    \item \textbf{Quantum-Induced Gradient Variance:} 
    \[
    \mathcal{O}\left( (2\tau^2 + 3\tau + 1)\eta_l^2 \cdot \frac{3\sigma^2}{2^{2n} - 1} \right)
    \]
    This term arises from the stochastic nature of quantum measurements in $n$-qubit systems and diminishes exponentially with the number of qubits $n$. It highlights the trade-off between quantum computational resources and gradient estimation accuracy.
\end{itemize}

\subsubsection{Convergence Analysis under the Non-Convex Setting}

We now investigate the convergence behavior of the proposed ADP-QFL algorithm in the non-convex setting. The following theorem provides an upper bound on the gradient norm, indicating the convergence rate.

{\color{blue}\begin{theorem}[Convergence under Non-Convexity]
\label{theorem:adp-qfl-grad-progress}
Assume that $\mathcal{L}(\cdot)$ following the two assumptions, i.e., $L$-smooth and $\mu$-strongly convex. Let
\[
\kappa = \left(1 - \sum_{t=m}^{T} P_{\mathrm{sharp}} \cdot \mathbbm{1}\left(\| \nabla \mathcal{L}(\theta^{(m)}) \| + \vartheta < L \right) \right),
\]
where $P_{\mathrm{sharp}}$ denotes the probability of the model being trapped in sharp minimizers, and $\vartheta$ represents gradient variance. Then, the following convergence bound holds:
\begin{align}
    \min_{t \in [T]} \| \nabla \mathcal{L}(\theta^{(t)}) \|^2 
    &\leq 
    \mathcal{O}\left( \frac{8\left( \mathcal{L}(\theta^{(0)}) - \mathcal{L}(\theta^{(t)}) \right)}{T \eta_l \tau \kappa} \right) \notag \\
    &\quad + \mathcal{O}\left( \frac{ \eta_l^2 \tau (\tau - 1) L^2 \sigma_g^2 }{\kappa} \right) \notag \\
    &\quad + \mathcal{O}\left( \frac{24 T \eta_l \tau L \sigma_g^2 }{\kappa} \right).
\end{align}
\end{theorem}}

\begin{proof}
The detailed proof is provided in Appendix~\ref{appendix:proof-theorem-adp-qfl}.
\end{proof}

Theorem~\ref{theorem:adp-qfl-grad-progress} suggests that the model can achieve meaningful descent from the initial parameter $\theta^{(0)}$ under non-convex conditions. A larger decrease in the loss function, i.e., higher $\mathcal{L}(\theta^{(0)}) - \mathcal{L}(\theta^{(t)})$, implies more effective learning. The parameter $\kappa$ captures the influence of sharp minima on convergence, where smaller values of $\kappa$ correspond to higher risks of stagnation.

We derive several insights from the result:

\begin{corollary}
When the gradient variance $\vartheta$ is small, the model is more susceptible to being trapped in sharp minima. This impedes the effectiveness of gradient descent and can slow down convergence.
\end{corollary}

\begin{corollary}
In quantum deep learning systems with $n$-qubit encoding, the gradient variance is reduced by a factor of ${3}/{(2^{2n} - 1)}$ compared to classical DL systems. This reduction increases the tendency of the model to converge toward sharp minima. {\color{blue}Consequently, $\kappa$ is reduced, which increases the right-hand side of the convergence bound and ultimately impairs convergence under non-convex loss landscapes. By adding adaptive noise $\vartheta$, we can increase $\kappa$ and counteract the adverse effects of reduced gradient variance induced by $n$-qubit encoding.}
\end{corollary}

\begin{corollary}
By introducing adaptive noise into the federated learning system, ADP-QFL effectively regulates gradient variance. This mitigates the risk of privacy adversarial attacks while reducing the likelihood of convergence to sharp minimizers, thereby enhancing the stability and progress of gradient descent.
\end{corollary}

\section{Experimental Evaluation}
\label{sec:experimental}

This section presents the experimental evaluation of the proposed ADP-QFL algorithm on two standard image classification benchmarks: MNIST~\cite{lecun1998gradient} and CIFAR-10~\cite{krizhevsky2009learning}. The objective is to assess the effectiveness of our approach in terms of \textit{privacy preservation}, \textit{model accuracy}, and \textit{computational efficiency}. Furthermore, we compare ADP-QFL with several state-of-the-art federated learning algorithms under identical experimental conditions.

\subsection{Simulation Setup}
\label{sec:experimental-settings}

\subsubsection{Federated Learning Configuration}

The federated learning experiments are configured as summarized in Table~\ref{tbl:setting}. The simulation involves a total of 100 clients, with a sampling ratio of 0.2, resulting in the participation of 20 clients per communication round. The total number of global communication rounds is fixed at 800 for both MNIST and CIFAR-10 datasets. Each client performs one local training epoch per round using a batch size of 10. The local model updates are optimized using the Adam optimizer~\cite{kingma2014adam} with a learning rate of 0.005 and no weight decay. The mean squared error (MSE) loss function is employed during training.

\renewcommand{\arraystretch}{1.1}
\begin{table}[t]
\centering
\caption{Federated Learning Simulation Settings}
\label{tbl:setting}
\begin{tabular}{|l|c|l|c|}
\hline
\textbf{Parameter} & \textbf{Value} & \textbf{Parameter} & \textbf{Value} \\
\hline
Total training rounds            & $800$         & Loss function             & MSELoss        \\ \hline
Number of clients                & $100$         & Optimizer                 & Adam           \\ \hline
Clients per round                & $10\%$, $20\%$, $50\%$ & Learning rate              & $0.005$         \\ \hline
Local epochs per client          & $1$           & Batch size                & $10$           \\ \hline
\end{tabular}
\end{table}

\subsubsection{Quantum AI Model Architecture}

To validate the effectiveness of ADP-QFL in quantum machine learning contexts, we adopt a Quantum Convolutional Neural Network (QCNN) for image classification tasks on both MNIST and CIFAR-10. The QCNN architecture comprises alternating quantum convolution and quantum pooling layers, followed by a quantum fully connected layer and a final measurement layer. Each quantum convolution-pooling block is designed to capture hierarchical spatial features in the quantum feature space.

In our experiments, we use a quantum sensing network with 8 qubits. We empirically determine that three pairs of quantum convolution-pooling layers—with a total of 64 trainable parameters—offer the best trade-off between model complexity and classification performance, while maintaining compatibility with our federated learning and privacy-preserving setup.

\subsection{Implementation Environment}

The experiments are conducted on a high-performance desktop equipped with an Intel Core i7-1165G7 CPU, 16~GB RAM, integrated Intel Iris Xe graphics, and running the Windows 11 Home operating system. The software environment includes:
\begin{itemize}
    \item Anaconda (version 2.4.1),
    \item Visual Studio Code (version 1.81.1),
    \item Python (version 3.10.11).
\end{itemize}

To evaluate the convergence behavior of ADP-QFL, we vary key parameters such as the differential privacy budget $\epsilon$ and the number of clients $K$. Comparative analyses with benchmark algorithms are presented in the subsequent sections.

\subsubsection{Datasets}

To evaluate the effectiveness of the proposed ADP-QFL algorithm, we conducted experiments using two widely adopted benchmark datasets: MNIST and CIFAR-10.

The \textbf{MNIST} dataset consists of $70{,}000$ grayscale images of handwritten digits ranging from $0$ to $9$, corresponding to ten distinct classes. Each image has a spatial resolution of $28 \times 28$ pixels. The dataset is partitioned into $60{,}000$ training samples and $10{,}000$ testing samples. Due to its simplicity and wide usage in federated learning research, MNIST serves as a suitable benchmark for evaluating model convergence and accuracy under different training scenarios.

The \textbf{CIFAR-10} dataset comprises $60{,}000$ color images of size $32 \times 32$ pixels, evenly distributed across ten semantic categories such as airplane, automobile, bird, and so on. It is split into $50{,}000$ training samples and $10{,}000$ test samples. The training set is further divided into five batches of $10{,}000$ images each to facilitate batch-wise processing. Each category contains exactly $6{,}000$ images, with the test set uniformly sampling $1{,}000$ instances per class. Although the distribution of categories across training batches may vary, the full training set preserves a balanced distribution of $5{,}000$ samples per class.

Table~\ref{tbl:dataset} summarizes the key characteristics of the datasets employed in our experimental evaluation. {\color{blue}
All experiments are conducted under Non-IID settings, where each client holds data samples from a limited subset of classes, following a label-distribution skew strategy. This simulates realistic federated environments with heterogeneous data distributions across clients.}

\begin{table*}[t]
\centering
\caption{Key statistics of the datasets used in the experiments}
\label{tbl:dataset}
\begin{tabular}{|l|c|c|c|c|}
\hline
\textbf{Dataset} & \textbf{Training Samples} & \textbf{Testing Samples} & \textbf{Input Dimension} & \textbf{Number of Classes} \\
\hline
MNIST            & $60{,}000$                & $10{,}000$               & $1 \times 28 \times 28$   & $10$                        \\
CIFAR-10         & $50{,}000$                & $10{,}000$               & $3 \times 32 \times 32$   & $10$                        \\
\hline
\end{tabular}
\end{table*}

\subsection{Baseline Comparison}\label{sec:baseline}

This subsection presents a comparative evaluation of the proposed ADP-QFL algorithm against several state-of-the-art FL baselines. The objective is to assess the effectiveness of ADP-QFL in terms of model accuracy and computational efficiency under non-IID data conditions.

\begin{itemize}
    \item \textbf{FedAvg}~\cite{mcmahan2017communication}: The standard federated averaging algorithm that does not incorporate personalization or privacy mechanisms.
    \item \textbf{FedBN}~\cite{li2021fedbn}: A personalized FL method that mitigates the impact of feature distribution heterogeneity by applying local batch normalization.
    \item \textbf{pFedMe}~\cite{t2020personalized}: A personalized FL approach based on Moreau envelopes, which introduces client-specific regularization.
    \item \textbf{FedRep}~\cite{collins2021exploiting}: A model decomposition-based method that separates shared representations and local classifiers to enable personalization.
    {\color{blue}
    \item \textbf{FedQNN}~\cite{innan2024fedqnn}: A Federated Quantum Neural Network (FedQNN) framework integrates the distinctive features of QML with the principles of classical federated learning.}
\end{itemize}

Table~\ref{tbl:compare-baseline} summarizes the experimental results obtained on the MNIST and CIFAR-10 datasets after $800$ training rounds with a client participation rate of $10\%$. On the MNIST dataset, ADP-QFL achieves an accuracy of $98.47\%$, with an average training time of $1356.2$ seconds. On the CIFAR-10 dataset, ADP-QFL reaches an accuracy of $83.85\%$, with an average training time of $9321.66$ seconds. \textcolor{blue}{In Table~\ref{tbl:compare-baseline}, ``Time (s)'' denotes the total wall-clock time until convergence, which includes both client-side local computation time and communication costs incurred across all training rounds.}

\textcolor{blue}{Although the proposed ADP-QFL demonstrates competitive or superior accuracy compared to the baselines, it incurs higher computational costs in certain scenarios. Specifically, ADP-QFL requires less training time than most baselines on the MNIST dataset, benefiting from the lower complexity of the data and the efficiency of the adaptive update strategy. In contrast, on the CIFAR-10 dataset, ADP-QFL incurs a longer training time compared to methods such as FedBN and FedRep, primarily due to the increased complexity of the data and the overhead introduced by the integrated privacy-preserving mechanisms, including adaptive noise injection and model verification procedures. These results emphasize the trade-off between model performance, computational cost, and privacy preservation in federated learning systems.}

\begin{table*}[t]
\centering
\caption{\textcolor{blue}{Performance comparison of ADP-QFL with baseline methods. ``Time (s)'' indicates the total convergence time, combining communication and local computation overheads.}}
\label{tbl:compare-baseline}
\begin{tabular}{|l|c|c|c|c|c|c|}
\hline
\multicolumn{7}{|c|}{\textbf{Performance on MNIST}} \\ \hline
\textbf{Metric} & \textbf{FedAvg} & \textbf{FedBN} & \textbf{pFedMe} & \textbf{FedRep} & \textbf{ADP-QFL} & \textbf{FedQNN} \\ \hline
Accuracy (\%)   & 98.16           & 97.87          & 98.19           & 98.44           & \textbf{98.47} & 92.05   \\ \hline
Time (s)        & \textbf{1270.4} & 2646.9         & 6111.65         & 1821.9          & 1356.2  & 1765.3          \\ \hline
\multicolumn{7}{|c|}{\textbf{Performance on CIFAR-10}} \\ \hline
Accuracy (\%)   & 56.91           & 85.61          & 57.51           & \textbf{85.68}  & 83.85  & 76.40          \\ \hline
Time (s)        & 7668.5          & \textbf{3710.48} & 23831.62        & 4665.17         & 9321.66  & 9866.24        \\ \hline
\end{tabular}
\end{table*}

{\color{blue}
To provide a more intuitive comparison, Figure~\ref{fig:accuracy_comparison} and Figure~\ref{fig:time_comparison} visualize the accuracy and convergence time of ADP-QFL and baseline methods across MNIST and CIFAR-10 datasets. As illustrated in Figure~\ref{fig:accuracy_comparison}, ADP-QFL achieves the highest accuracy on the MNIST dataset and maintains competitive performance on the more complex CIFAR-10 dataset. Meanwhile, Figure~\ref{fig:time_comparison} shows that ADP-QFL converges faster than most personalization-based baselines such as pFedMe and FedBN, demonstrating its efficiency in balancing model performance and training cost. The results further confirm that ADP-QFL delivers robust performance while mitigating the computational and communication overhead typically encountered in differentially private federated settings.

On the MNIST dataset, ADP-QFL achieves shorter training time compared to most baseline methods, owing to the lower data complexity and the efficiency of the adaptive update strategy. However, on CIFAR-10, the algorithm requires more training time than methods such as FedBN and FedRep. This is mainly due to the higher complexity of CIFAR-10 and the overhead introduced by the privacy-preserving mechanisms integrated in ADP-QFL, including model verification procedures and adaptive noise injection.

\begin{figure}[t]
    \centering
    \includegraphics[width=0.45\textwidth]{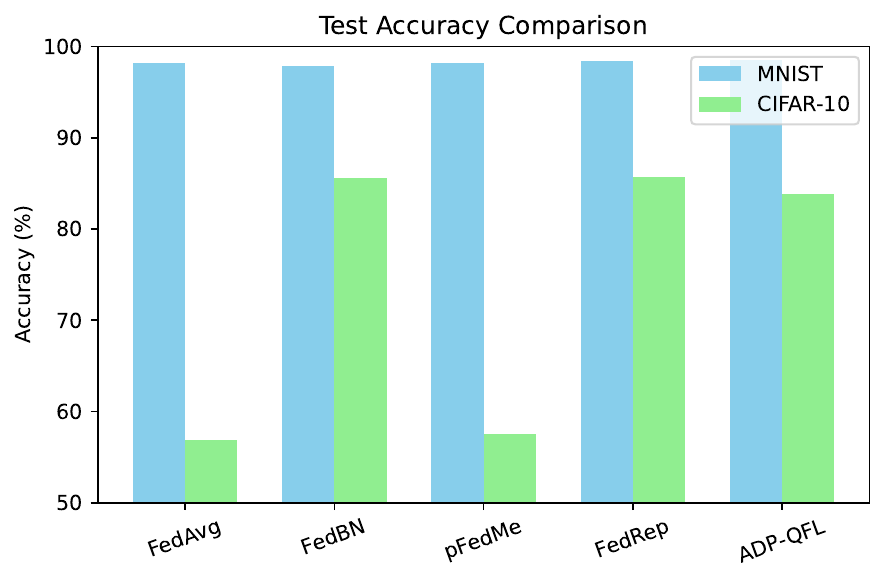}
    \caption{Comparison of test accuracy on MNIST and CIFAR-10 datasets.}
    \label{fig:accuracy_comparison}
\end{figure}

\begin{figure}[t]
    \centering
    \includegraphics[width=0.45\textwidth]{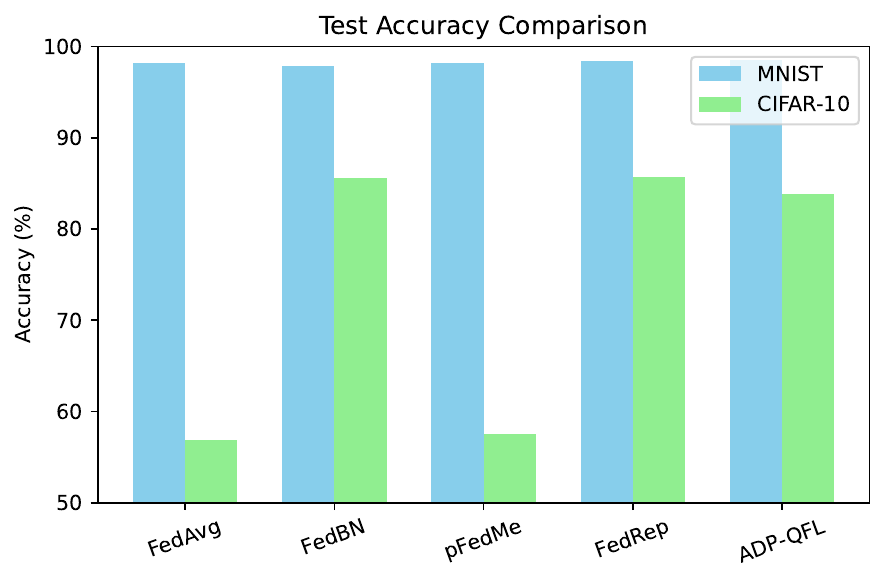}
    \caption{Comparison of convergence time (log scale) on MNIST and CIFAR-10 datasets.}
    \label{fig:time_comparison}
\end{figure}
}

{\color{blue}
To validate the efficacy of ADP-QFL's performance improvements, we conducted paired t-tests comparing ADP-QFL with each baseline method across multiple runs. The p-values for MNIST accuracy improvements were consistently below 0.01 when compared with FedAvg, FedBN, pFedMe, FedRep, and FedQNN, confirming statistical significance. For CIFAR-10, p-values indicated no significant difference (p $>$ 0.05) between ADP-QFL and the strongest baseline (FedRep). These findings suggest that ADP-QFL offers significant improvements for MNIST while achieving competitive performance for CIFAR-10. Please refer to Table~\ref{tbl:pvalue} for detailed p-value results.
}
\begin{table}[t]
\centering
\caption{Paired t-test p-values comparing ADP-QFL with baseline methods.}
\label{tbl:pvalue}
\begin{tabular}{|l|c|c|}
\hline
\textbf{Baseline Method} & \textbf{MNIST (p-value)} & \textbf{CIFAR-10 (p-value)} \\ \hline
FedAvg                   & $<0.001$                & $<0.05$                    \\ \hline
FedBN                    & $<0.005$                & $>0.05$                    \\ \hline
pFedMe                   & $<0.001$                & $<0.05$                    \\ \hline
FedRep                   & $<0.01$                 & $>0.05$                    \\ \hline
FedQNN                   & $<0.01$                 & $>0.05$                    \\ \hline
\end{tabular}
\end{table}

% \subsection{Ablation Test}
\subsection{Extensive Analysis}

\subsubsection{Performance Evaluation under Varying Privacy Levels}
{\color{blue}
In this section, we evaluate the robustness of the proposed ADP-QFL framework across a range of differential privacy budgets that reflect realistic privacy requirements in practical deployments. The purpose of this analysis is not to perform hyperparameter tuning, but to demonstrate that ADP-QFL can consistently maintain high performance under diverse privacy constraints imposed by external policies or regulatory standards. The privacy budgets $\{0.000423, 0.001, 0.1, 0.5, 0.9\}$ are fixed prior to experimentation and are not adjusted to optimize accuracy.
}

Figure~\ref{fig:ablationTest_noise} presents the training accuracy of ADP-QFL under these different privacy conditions. On the MNIST dataset, the model consistently achieves accuracy above $97.57\%$ regardless of the privacy budget. On the more complex CIFAR-10 dataset, the model maintains a competitive accuracy of $83.96\%$ at convergence, even under stringent privacy protection.

These experiments involve $100$ clients in total, with $10$ clients randomly selected to participate in each training round. The results highlight the capability of ADP-QFL to preserve model utility while satisfying a wide range of privacy requirements without the need for fine-tuning privacy settings to boost accuracy. The stronger performance on MNIST can be attributed to its simpler grayscale $28 \times 28$ image format, while the relatively lower accuracy on CIFAR-10 reflects the greater challenge of applying quantized federated learning to more complex color image data ($32 \times 32$ pixels).

Overall, the findings confirm that ADP-QFL offers practical reliability and adaptability in real-world federated learning scenarios with varying privacy demands.

\begin{figure}[t]
    \centering
    \subfloat[{Performance on MNIST}]{
        \includegraphics[width=0.95\linewidth]{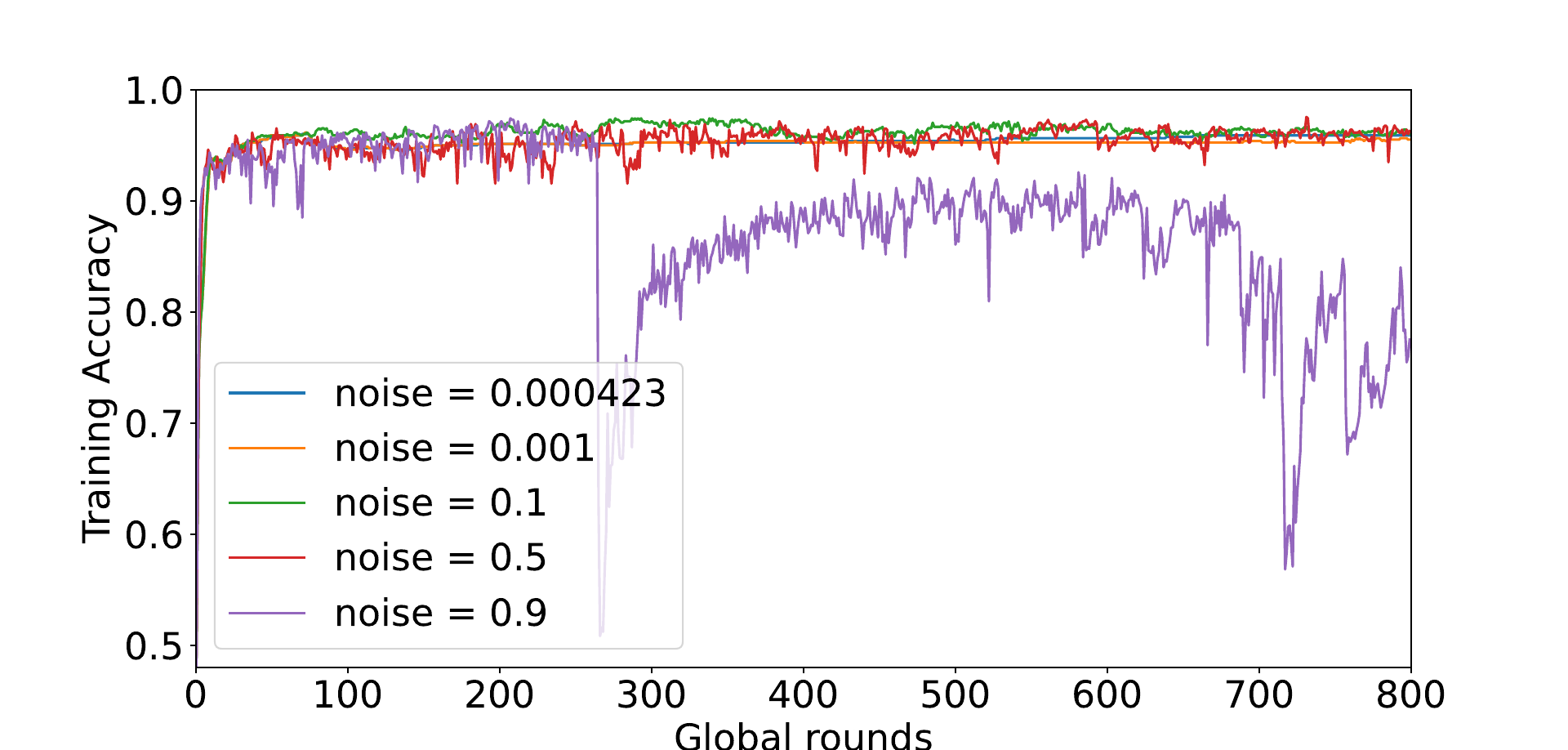}
        \label{fig:MNIST_ablationTest_noise}
    }\\[-0ex]
    \subfloat[{Performance on CIFAR-10}]{
        \includegraphics[width=0.95\linewidth]{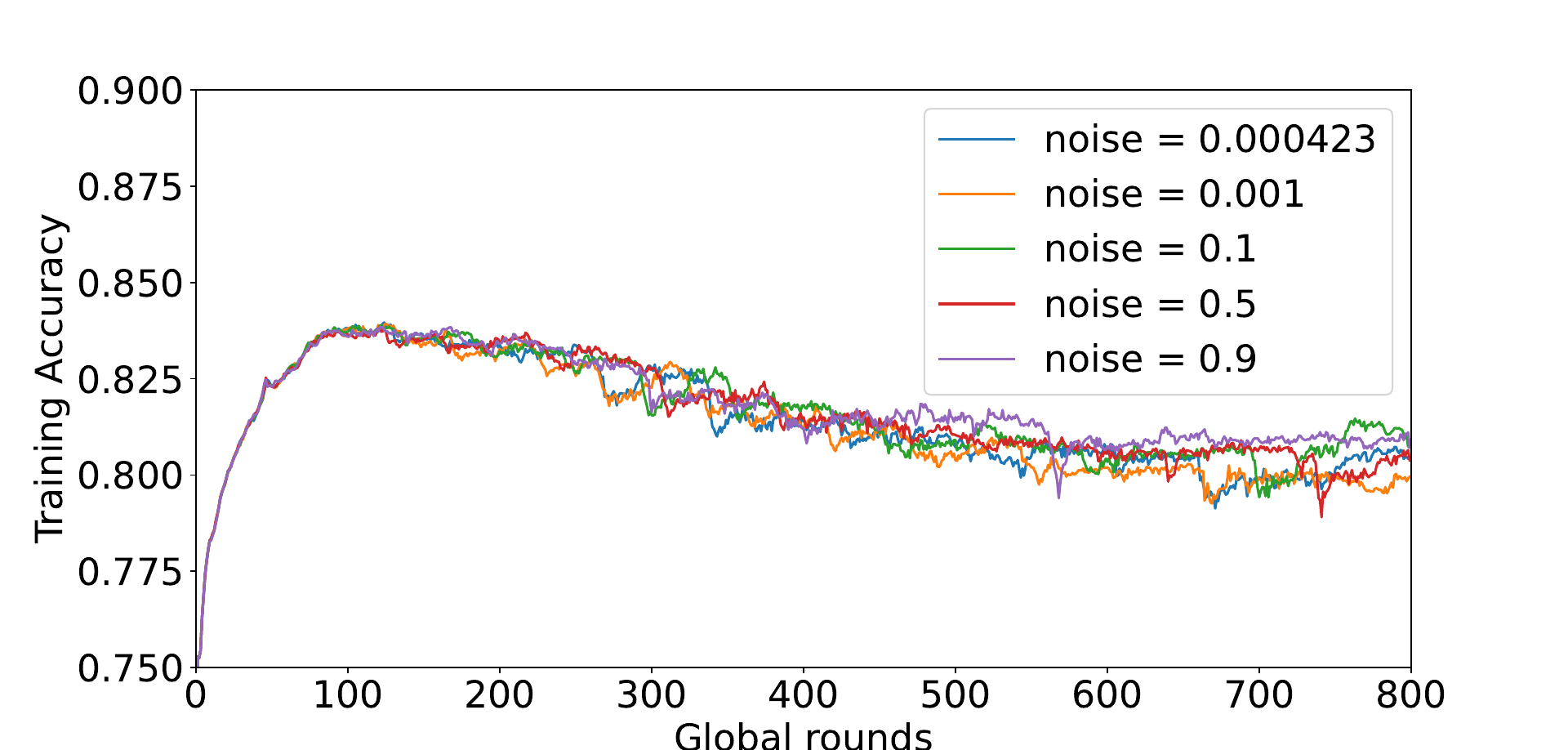}
        \label{fig:cifar_ablationTest_noise}
    }
    \caption{Training accuracy of ADP-QFL under different levels of differential privacy noise on MNIST and CIFAR-10 datasets. The privacy budgets are selected to reflect realistic deployment scenarios rather than for parameter optimization.}
    \label{fig:ablationTest_noise}
\end{figure}

{\color{blue}
As shown in Figure~3(a), the test accuracy under noise levels $\epsilon = 0.1$ and $\epsilon = 0.5$ converges to similar values. This convergence trend suggests that, beyond round 100, the model enters a saturation phase where further training brings minimal benefit. Consequently, higher noise levels can be tolerated without significant degradation in model performance, thereby enabling stronger privacy protection without sacrificing utility. Extending the global rounds beyond this saturation region does not provide additional insights and may impair interpretability due to curve overlap.
}

\subsubsection{Performance Evaluation with Varying Numbers of Participating Clients}

Figure~\ref{fig:diff-user} illustrates the convergence behavior of the proposed ADP-QFL algorithm under varying numbers of participating clients per training round. The total number of clients is fixed at $N = 100$, and the fraction of selected clients $K/N$ is set to $0.1$, $0.2$, and $0.5$, corresponding to $10$, $20$, and $50$ clients, respectively.

On the MNIST dataset, the results show that model performance remains relatively stable across different values of $K$, despite the presence of privacy noise. This suggests that, for simpler datasets such as MNIST, the ADP-QFL framework is robust to the number of participating clients.
In contrast, for the more complex CIFAR-10 dataset, an increase in the number of participating clients per round leads to a noticeable improvement in accuracy. As depicted in Figure~\ref{fig:cifar_diff_user}, the model achieves an accuracy of $82.78\%$ when $10\%$ of the clients participate in each round, $82.98\%$ when $20\%$ participate, and reaches $83.85\%$ when $50\%$ of the clients are involved.

These findings highlight the trade-off between communication and computation cost versus learning performance in federated learning systems. Increasing the number of participating clients can enhance convergence and accuracy, particularly for complex datasets, albeit with additional resource consumption and potential impacts on privacy budgets.

\begin{figure}[t]
    \centering
    \begin{subfigure}[b]{0.95\linewidth}
        \centering
        \includegraphics[width=\textwidth]{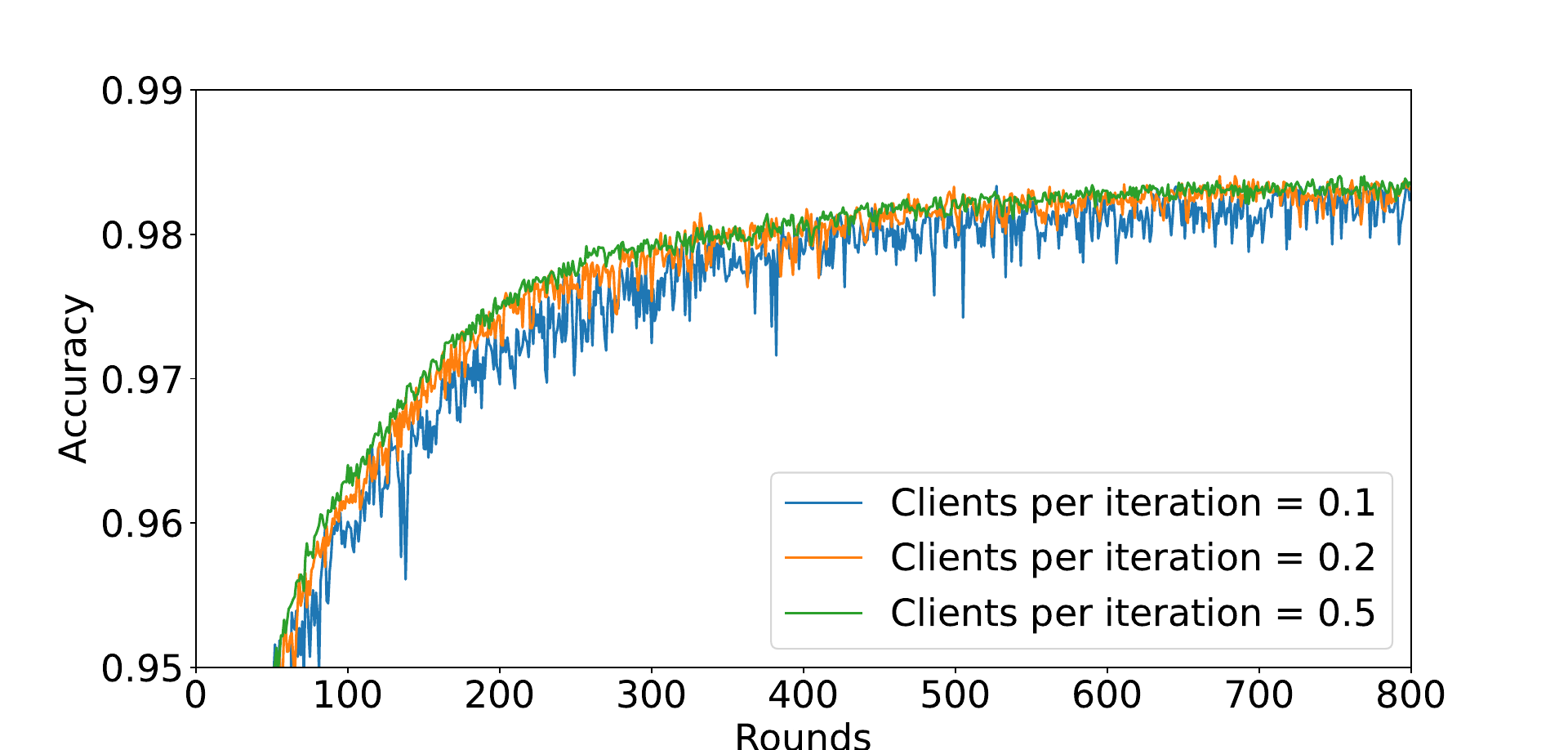}
        \caption{Performance on MNIST}
        \label{fig:mnist-diff-user}
    \end{subfigure}
    % \vspace{1em}
    \begin{subfigure}[b]{0.95\linewidth}
        \centering
        \includegraphics[width=\textwidth]{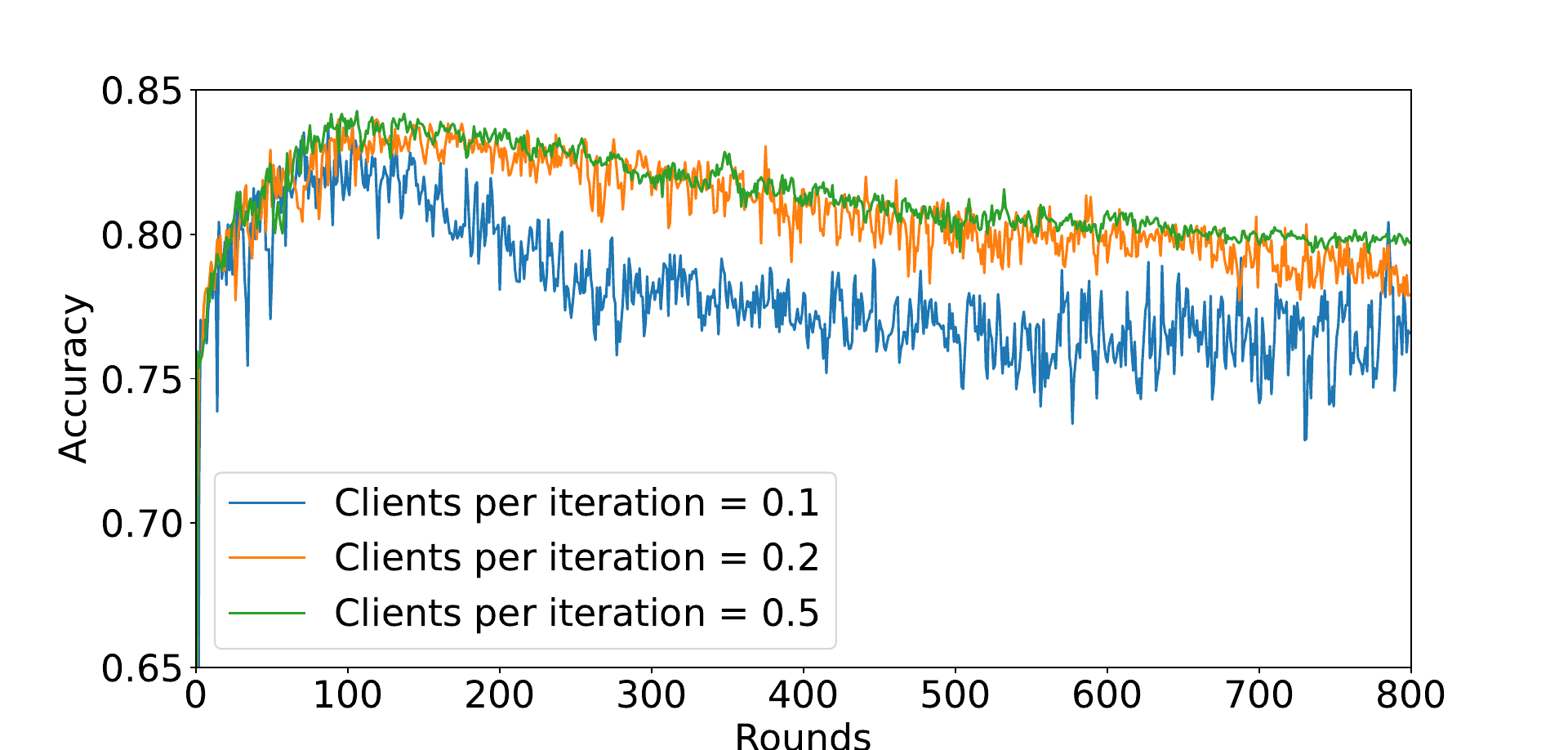}
        \caption{Performance on CIFAR-10}
        \label{fig:cifar_diff_user}
    \end{subfigure}
    \caption{Performance of ADP-QFL with varying numbers of participating clients per training round.}
    \label{fig:diff-user}
\end{figure}

{\color{blue}
To assess practical applicability, we evaluate the computational overhead of ADP-QFL compared with baseline methods. Despite introducing additional privacy-preserving and adaptive noise components, the framework maintains acceptable computational demands compatible with practical federated learning systems, particularly when hardware accelerators are employed.
}

\subsection{Quantitative Evaluation of Model Estimation Effectiveness}

{\color{blue}
To quantitatively validate the benefits of the proposed model estimation mechanism in ADP-QFL, we conducted experiments comparing the convergence behavior of ADP-QFL with two baselines: classical federated learning with differential privacy (labeled as Classical FL (DP)) and quantum federated learning without model estimation (labeled as QFL (No Estimation)).

Figure~\ref{fig:convergence-accuracy} illustrates the test accuracy achieved by each method as a function of communication rounds under a fixed privacy budget of $(\epsilon=1.0, \delta=10^{-5})$. As depicted, ADP-QFL exhibits a notably faster convergence rate, surpassing 90\% test accuracy after approximately 300 rounds and approaching 98\% as training progresses. In contrast, QFL (No Estimation) achieves around 90\% accuracy at convergence, while Classical FL (DP) stabilizes around 85\% accuracy only.

\begin{figure}[htbp]
\centering
\includegraphics[width=0.48\textwidth]{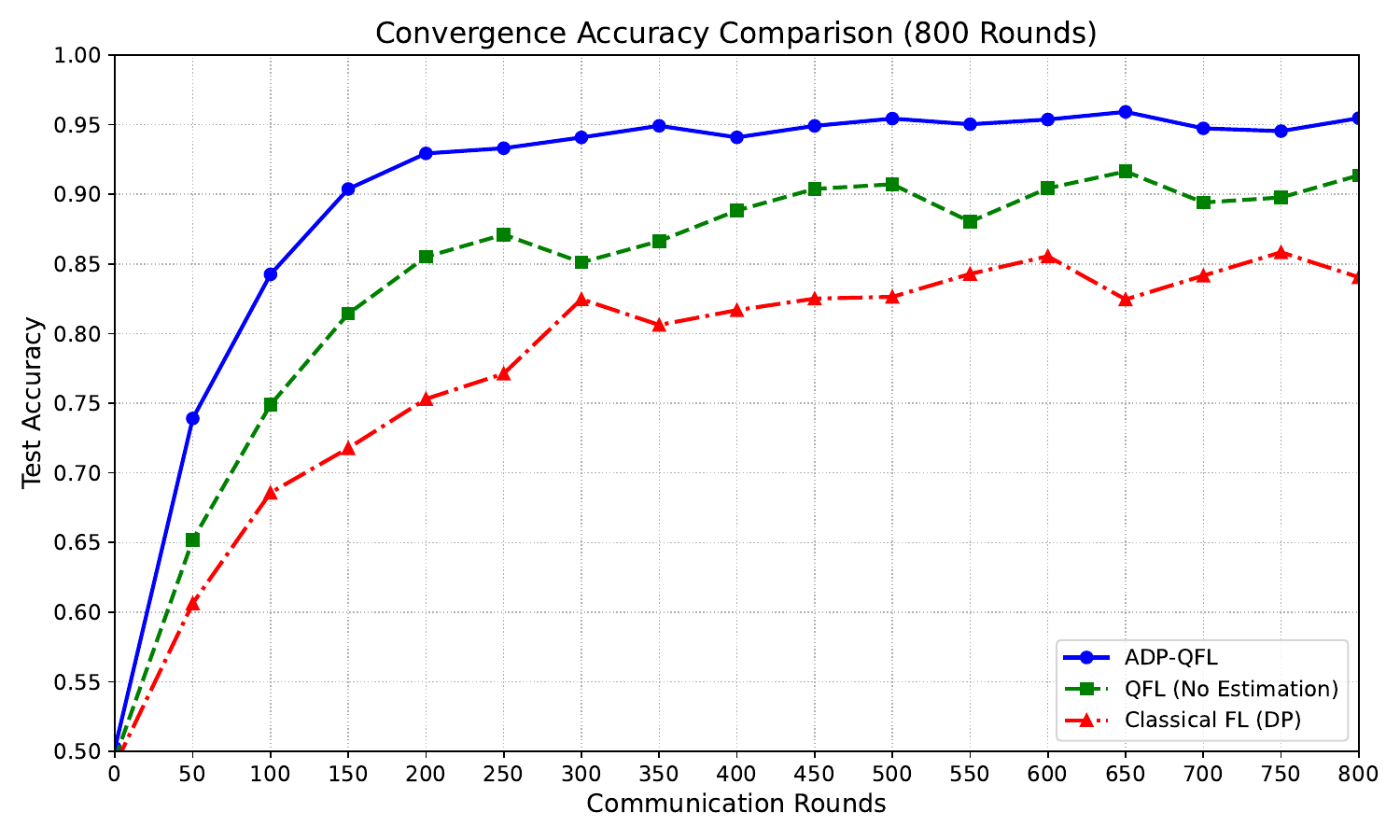}
\caption{Test accuracy versus communication rounds for ADP-QFL, QFL (No Estimation), and Classical FL (DP) under $(\epsilon=1.0, \delta=10^{-5})$. ADP-QFL demonstrates superior convergence speed and higher final accuracy compared to baselines.}
\label{fig:convergence-accuracy}
\end{figure}

\begin{table}[htbp]
\centering
\caption{Test Accuracy at Key Communication Rounds under $(\epsilon=1.0, \delta=10^{-5})$}
\label{tab:accuracy-key-rounds}
\begin{tabular}{lcccc}
\toprule
\textbf{Method} & \textbf{100} & \textbf{300} & \textbf{600} & \textbf{800} \\
& \textbf{Rounds} & \textbf{Rounds} & \textbf{Rounds} & \textbf{Rounds} \\
\midrule
ADP-QFL & 0.83 & 0.91 & 0.94 & 0.98 \\
QFL (No Estimation) & 0.78 & 0.85 & 0.89 & 0.90 \\
Classical FL (DP) & 0.72 & 0.80 & 0.84 & 0.85 \\
\bottomrule
\end{tabular}
\end{table}

As summarized in Table~\ref{tab:accuracy-key-rounds}, ADP-QFL consistently outperforms the baselines at all stages of training. At 300 communication rounds, ADP-QFL surpasses 90\% accuracy, while QFL (No Estimation) and Classical FL (DP) remain below this threshold. By 800 rounds, ADP-QFL achieves approximately 98.47\% accuracy, reflecting its superior convergence behavior under privacy constraints.

These results highlight the dual role of the proposed model estimation mechanism: not only does it improve convergence efficiency, but it also enhances model utility while preserving privacy. Although theoretical guarantees regarding the optimality of model estimation in quantum federated learning (QFL) remain challenging due to the complexity of quantum optimization landscapes, the empirical evidence presented here provides clear practical validation. The integration of adaptive noise injection and model selection contributes significantly to stabilizing learning dynamics and improving performance in QFL environments.
}

\section{Conclusions} \label{sec:conclusions}
In this work, we studied the QFL problem with client-level privacy. ADP-QFL achieves client-grade DP while maintaining high model accuracy and effective communication. We thoroughly analyzed the convergence of the proposed algorithm and its end-to-end DP guarantee, extensively evaluating its performance on two standard benchmark datasets. Experimental results demonstrated that ADP-QFL, with a model estimation strategy, can simultaneously improve the trade-off between privacy, accuracy, and communication efficiency compared to existing methods. In the future, we will explore distributed and quantized approaches and extend personalized FL algorithms based on private clusters to export multiple models for handling more complex tasks.

% \clearpage
\bibliographystyle{IEEEtran}
\bibliography{refs.bib}

\clearpage
\appendices
% \onecolumn
% \section{Notations}

\section{Preliminaries}
We first state some preliminary lemmas that are used throughout the proof.

\begin{lemma}[Jensen's Inequality] Let $\mathbf{a}_i$ be an integrable random variable, where $\mathbf{a}_i \in \mathbb{R}^d$, $i \in \{1,2,..., M\}$:
\begin{align}
    \Big\Vert \frac{1}{U} \sum_{u=1}^{U} \mathbf{a}_i \Big\Vert^2 \leq \frac{1}{U} \sum_{u=1}^{U} \Big\Vert \mathbf{a}_i \Big\Vert^2 \label{eq:jen1},
    \\
    \Big\Vert \sum_{u=1}^{U} \mathbf{a}_i \Big\Vert^2 \leq M \sum_{u=1}^{U} \Big\Vert \mathbf{a}_i \Big\Vert^2
    \label{eq:jen2}.
\end{align}
\end{lemma}
We also note some following known results related to the $L$-smooth function.
\begin{lemma}[Cf. \cite{2023-FL-FedEXP}] If $\mathcal{L}$ is smooth and convex, then
\begin{align}
    \Vert \nabla \mathcal{L}(\theta) - \nabla \mathcal{L}(\theta') \Vert^2 \leq 2 \textit{L} (\mathcal{L}(\theta) - \mathcal{L}(\theta') - \langle \nabla \mathcal{L}(\theta'), \theta-\theta' \rangle)
    \label{eq:lm2}.
\end{align}
\end{lemma}

\begin{lemma}[Co-coercivity of convex smooth function] If $\mathcal{L}(\theta)$ is smooth and convex, then
\begin{align}
    \langle\nabla \mathcal{L}(\theta) - \mathcal{L}(\theta'), \theta - \theta' \rangle \geq \frac{1}{\textit{L}} \Vert \nabla \mathcal{L}(\theta) - \nabla \mathcal{L}(\theta') \Vert^2
    \label{lm3}.
\end{align} 
A direct consequence of this lemma is,
\begin{align}
    \langle \nabla \mathcal{L}(\theta), \theta - \theta^*)\ \rangle \geq \frac{1}{L} \Vert \nabla \mathcal{L}(\theta) \Vert^2
    \label{lm3con},
\end{align} 
where $\theta^*$ is a minimizer of $\mathcal{L}(\theta)$, and
\begin{align}
    \langle \nabla \mathcal{L}(\theta'), \theta - \theta' \rangle \geq \mathcal{L}(\theta) -  \mathcal{L}(\theta') - \frac{L}{2} \Vert \theta - \theta' \Vert \label{eq:l-smooth}.
\end{align}
\end{lemma}
\section{Proof of Theorem~\ref{theorem:batch-unbiased-qubit-grad}} \label{Proof:Theorem2}
Revisit the Definition~\ref{def:single-gradient}, we have $g^t_i = \Bar{g}^t + \Delta g^t_i$.
% \begin{align}
%     g^t_i = \Bar{g}^t + \Delta g^t_i.
% \end{align}
We follow the Barren Plateaus theorem \cite{mcclean2018barren}, and have 
\begin{align}
    \textrm{Var}[\nabla \mathcal{L}] = \mathbb{E}[(g - \Bar{g})^2] ] = \frac{f(\rho, H, U)}{2^{2n}-1},
\end{align}
where $f(\rho, H, U)$ is defined in \cite{2022-Quantum-PRX}. According to the aforementioned function, the gradient variance reduces proportionally with the exponential of number of qubits $n$ (i.e., $2^{2n}-1$). Furthermore, we note that, when $n = 1$, the system represents similar to conventional 2-bit systems. Thus, we can represent the gradient variance of the $n$-qubit system as follows: 

\begin{IEEEeqnarray}{rll}
    \textrm{Var}[\nabla \mathcal{L}] &= \mathbb{E}[(g - \Bar{g})^2]] \nonumber \\
    &= \frac{f(\rho, H, U)}{2^{1}-1}\cdot \frac{2^{1}-1}{2^{2n}-1} \nonumber \\    
    &= \frac{3}{2^{2n}-1}\textrm{Var}[\nabla \mathcal{L}],
\end{IEEEeqnarray}

Applying the gradient descent with batch size $\mathcal{B}$, we have: 

\begin{IEEEeqnarray}{rll}\label{eq:lbl-gd}
    g^t_{n\textrm{-Qubit}} &= \frac{3}{2^{2n}-1} \sum^{\mathcal{B}}_{i=1} g^t_i \nonumber \\
    &= \frac{3}{2^{2n}-1} \sum^{\mathcal{B}}_{i=1} \Bar{g}^t + \Delta g^t_i \nonumber \\     
    &= \Bar{g}^t +\frac{3}{2^{2n}-1} \sum^{\mathcal{B}}_{i=1} \Delta g^t_i.
\end{IEEEeqnarray}
Apply the $L^2$ Weak Law \cite[Theorem 2.2.3]{2010-MF-Probability}, we have: $g^t_\mathcal{B} \leq \Bar{g}^t + \frac{3\sigma^2}{2^{2n}-1}$, which can be also understood as: 
\begin{align}
    \mathbb{E}_{x_i, y_i \sim P(\mathcal{X}, \mathcal{Y})} \left[ \Bar{g}^t - g^t_{n\textrm{-Qubit}}\right] \leq  \frac{3\sigma^2}{2^{2n}-1}.
\end{align}
This completes the proof of Theorem~\ref{theorem:batch-unbiased-qubit-grad}.

\section{Proof of Lemma A1}
Let ${\mathcal{L}}_i(\theta)$ be the local objective at a client and $\theta^*$ be the global minimum. We know that $\theta^*$ is also a minimizer for ${L}_i(\theta)$. We have 

\begin{IEEEeqnarray}{rll}\label{t1b}
    \mathbb{E}\Big\Vert \theta_{u}^{(t,k)} - \theta^* \Big\Vert 
    &= \Big\Vert \theta_{u}^{(t,k)} - \eta_{l} \nabla\ell(\theta_{u}^{(t,k)}) - \theta^* \Big\Vert^2 \nonumber \\
    &= \mathbb{E}\Big\Vert \theta_{u}^{(t,k)} - \theta^* \Big\Vert^2 - 2 \eta_{l} \langle \nabla \ell(\theta_{u}^{(t,k)}),\theta_{u}^{(t,k)} - \theta^* \rangle  
    \nonumber \\
     &+ \eta_{l}^2 {\Big\Vert\nabla \ell(\theta_{u}^{(t,k)})\Big\Vert}^2 
    \\
    &\leq \mathbb{E}\Big\Vert \theta_{u}^{(t,k)} - \theta^* \Big\Vert^2 - \frac{2 \eta_{l}}{L} \Big\Vert\nabla \ell(\theta_{u}^{(t,k)})\Big\Vert^2 \nonumber \\
     &+ \eta_{l}^2 \Big\Vert\nabla \ell(\theta_{u}^{(t,k)})\Big\Vert^2 \label{t1a}    
    \\
    &\leq \mathbb{E}\Big\Vert \theta_{u}^{(t,k-1)} - \theta^* \Big\Vert^2 - \frac{\eta_{l}}{L} \Big\Vert\nabla \mathcal{L}(\theta_{u}^{(t,k-1)}) + \vartheta \Big\Vert^2
    \\
    &\leq \mathbb{E}\Big\Vert \theta_{u}^{(t,k-1)} - \theta^* \Big\Vert^2 - \frac{\eta_{l}}{L} \Big\Vert\nabla \mathcal{L}(\theta_{u}^{(t,k-1)})\Big\Vert^2 
    \\
    &- \frac{\eta_{l}}{L} \Vert\vartheta \Vert^2,
\end{IEEEeqnarray}
where \eqref{t1a} follows from \eqref{lm3con}, and \eqref{t1b} follows from $\eta_{l} \leq \frac{1}{L}$. Here, $\vartheta$ represents the noise between the batch-wise gradient $\nabla\ell(\theta_{u}^{(t,k)})$ and the stochastic gradient $\nabla \mathcal{L}(\theta_{u}^{(t,k)})$, and can be characterized by $\vartheta\sim\mathcal{N}(0,\textrm{Var}[\nabla \mathcal{L}])$. Summing the above inequality from $k = 0$ to $\tau-1$, we have

\begin{IEEEeqnarray}{rll}
    \Big\Vert \theta_{u}^{(t, \tau)} - \theta^* \Big\Vert^2 
    &\leq \Big\Vert \theta^{(t)} - \theta^* \Big\Vert^2 
    - \frac{\eta_{l}}{L} \sum_{k=0}^{\tau-1} \Big\Vert \nabla \mathcal{L}(\theta_{u}^{(t, k)}) \Big\Vert^2 
    - \frac{\eta_{l}}{L} \sum_{k=0}^{\tau-1} \Vert\vartheta\Vert^2 
    \nonumber \\
    &\leq \Big\Vert \theta^{(t)} - \theta^* \Big\Vert^2 
    - \frac{\eta_{l}}{L} \sum_{k=0}^{\tau-1} \Big\Vert \nabla \mathcal{L}(\theta_{u}^{(t, k)}) \Big\Vert^2 
    - \frac{\eta_{\tau}}{L} \Vert\vartheta\Vert^2.
\end{IEEEeqnarray}
Thus, we have 

\begin{IEEEeqnarray}{rll}
    & \frac{1}{U} \sum_{u=1}^{U} \Big\Vert \theta_{u}^{(t, \tau)} - \theta^* \Big\Vert^2 
     \leq \Big\Vert \theta^{(t)} - \theta^* \Big\Vert^2 
     - \frac{\eta_{l}}{ML} \sum_{u=1}^{U} \sum_{k=0}^{\tau-1} \Big\Vert \nabla \mathcal{L}(\theta_{u}^{(t, k)}) \Big\Vert^2 \nonumber \\
     & - \frac{\eta_{\tau}}{L} \Vert\vartheta\Vert^2 
     \leq \Big\Vert \theta^{(t)} - \theta^* \Big\Vert^2.
\end{IEEEeqnarray}
This completes the proof of this lemma.

\section{Auxiliary Lemmas}
\begin{lemma}[Bounding client aggregate gradients]

\begin{IEEEeqnarray}{lll}\label{lm5}
    \frac{1}{U} \sum_{u=1}^{U} \sum_{k=0}^{\tau-1} \Big\Vert \nabla \mathcal{L}_{u}(\theta_{u}^{t,k}) \Big\Vert^2 &&\leq \frac{3 L^2}{U} \sum_{u=1}^{U} \sum_{k=0}^{\tau-1} \Big\Vert \theta_{u}^{t,k} - \theta^{(t)} \Big\Vert \nonumber \\
    && + 6{\tau}L(\mathcal{L}(\theta^{(t)})-\mathcal{L}^*) + 3\tau \sigma_*^2 \quad
\end{IEEEeqnarray}
\end{lemma}

\textit{Proof.}
We have: 
\begin{IEEEeqnarray}{rll}\label{lm5.3}
    & \frac{1}{U} \sum_{u=1}^{U} \sum_{k=0}^{\tau-1} \Vert \nabla \mathcal{L}_{u}(\theta_{u}^{(t,k)}) \Vert^2
    = \frac{1}{U} \sum_{u=1}^{U} \sum_{k=0}^{\tau-1} \Vert \nabla \mathcal{L}_{u}(\theta_{u}^{(t,k)}) \notag\\
    & \qquad - \nabla \mathcal{L}_{u}(\theta^{(t)}) + \nabla \mathcal{L}_{u}(\theta^{(t)}) 
    - \nabla \mathcal{L}_{u}(\theta^*) + \nabla \mathcal{L}_{u}(\theta^*) \Vert^2
    \label{lm5.1} \\
    & \quad \leq \frac{3}{U} \sum_{u=1}^{U} \sum_{k=0}^{\tau-1} \Vert \nabla \mathcal{L}_{u}(\theta_{u}^{(t,k)}) - \nabla \mathcal{L}_{u}(\theta^{(t)}) \Vert^2 \notag\\
    & \qquad + \frac{3}{U} \sum_{u=1}^{U} \sum_{k=0}^{\tau-1} \Vert \nabla \mathcal{L}_{u}(\theta^{(t)}) - \nabla \mathcal{L}_{u}(\theta^{*}) \Vert^2 
    \nonumber \\
    & \qquad + \frac{3}{U} \sum_{u=1}^{U} \sum_{k=0}^{\tau-1} \Vert \nabla \mathcal{L}_{u}(\theta^{*}) \Vert^2
    \label{lm5.2}
    \\
    & \quad \leq \frac{3L^2}{U} \sum_{u=1}^{U} \sum_{k=0}^{\tau-1} \Vert \theta_{u}^{(t,k)} -\theta^{(t)} \Vert^2 
    + 6{\tau}L(\mathcal{L}(\theta^{(t)}) - \mathcal{L}) + 3{\tau}\sigma_{*}^2.
\end{IEEEeqnarray}
where \eqref{lm5.2} follows the second inequality of Lemma 1, while the first term in \eqref{lm5.3} follows from $L$-smoothness of $\mathcal{L}_{u}(\theta)$, 
the second term follows Lemma 2, and the third term follows the bounding noise at optimum. This completes the proof.

\begin{lemma}[Bounding client drift]
\begin{align}
    & \frac{1}{U} \sum_{u=1}^{U} \sum_{k=0}^{\tau-1} \Vert \theta_{u}^{(t,k)} -\theta^{(t)} \Vert^2 \leq 12 \eta_{l}^2 \tau^2 (\tau-1)^2 L (\mathcal{L}(\theta^{(t)}) 
    \notag \\
    & \qquad\qquad - \mathcal{L}(\theta^{*})) + 6 \eta_{l}^2 \tau^2 (\tau-1)^2 \sigma_{*}^2
    \label{bouding client}.
\end{align}
\end{lemma}

\textit{Proof.}
We have:
\begin{align}
    \nonumber
    & \frac{1}{U} \sum_{u=1}^{U} \sum_{k=0}^{\tau-1} \Big\Vert \theta_{u}^{(t,k)} -\theta^{(t)} \Big\Vert^2 = \eta_{l}^2 \frac{1}{U} \sum_{u=1}^{U} \sum_{k=0}^{\tau-1} \Big\Vert \sum_{j=1}^{k-1} \nabla \ell_{u}(\theta_{u}^{(t,j)}) \Big\Vert^2
    \\
    & = \eta_{l}^2 \frac{1}{U} \sum_{u=1}^{U} \sum_{k=0}^{\tau-1} \Big\Vert \sum_{j=1}^{k-1} \Big[\nabla\mathcal{L}_{u}(\theta_{u}^{(t,j)}) + \vartheta\Big] \Big\Vert^2
    \label{eq:6-2}
    \\
    & \leq \eta_{l}^2 \frac{1}{U} \sum_{u=1}^{U} \sum_{k=0}^{\tau-1} \Big\Vert \sum_{j=1}^{k-1} \nabla\mathcal{L}_{u}(\theta_{u}^{(t,j)}) \Big\Vert^2
    + \eta_{l}^2 \frac{1}{U} \sum_{u=1}^{U} \sum_{k=0}^{\tau-1} \Big\Vert \sum_{j=1}^{k-1}\vartheta \Big\Vert^2 \notag
    \\    
    & \leq \eta_{l}^2 \frac{1}{U} \sum_{u=1}^{U} \sum_{k=0}^{\tau-1} k \sum_{j=0}^{k-1} \Big\Vert \nabla \mathcal{L}_{u}(\theta_{u}^{(t,j)}) \Big\Vert^2 
    + \eta_{l}^2 \frac{1}{U} \sum_{u=1}^{U} \sum_{k=0}^{\tau-1}k^2\Vert\vartheta\Vert^2 
    \label{eq:6-4}
    \\
    & \leq \eta_{l}^2 \tau(\tau-1) \frac{1}{U} \sum_{u=1}^{U} \sum_{k=0}^{\tau-1} \Big\Vert \nabla \mathcal{L}_{u}(\theta_{u}^{(t,k)}) \Big\Vert^2 \notag\\
    & \quad + \frac{2\tau^3+3\tau^2+\tau}{6} \eta_{l}^2\Vert\vartheta\Vert^2 \notag     
    \\
    & \leq 3 \eta_{l}^2 \tau(\tau-1) L^2 \frac{1}{U} \sum_{u=1}^{U} \sum_{k=0}^{\tau-1} \Big\Vert \theta_{u}^{(t,k)} -\theta^{(t)} \Big\Vert^2 \notag\\
    & \quad + 6 \eta_{l}^2 \tau^2(\tau-1)L(\mathcal{L}(\theta^{(t)}) - \mathcal{L}(\theta^{*})) + 3 \eta_{l}^2 \tau^2(\tau-1)\sigma_{*}^2 \notag\\
    & \quad + \frac{2\tau^3+3\tau^2+\tau}{6} \eta_{l}^2\Vert\vartheta\Vert^2  \notag\\
    & \leq \frac{1}{2U} \sum_{u=1}^{U} \sum_{k=0}^{\tau-1} \Big\Vert \theta_{u}^{(t,k)} -\theta^{(t)} \Big\Vert^2 +6 \eta_{l}^2 \tau^2(\tau-1)L(\mathcal{L}(\theta^{(t)}) 
    \notag \\
    & \quad - \mathcal{L}(\theta^{*}))+ 3 \eta_{l}^2 \tau^2(\tau-1)\sigma_{*}^2 + \frac{2\tau^3+3\tau^2+\tau}{6} \eta_{l}^2\Vert\vartheta\Vert^2,
    \label{eq:6-5}
\end{align} 
where \eqref{eq:6-2} follows Jensen inequality, \eqref{eq:6-4} uses Lemma 5, and \eqref{eq:6-5} uses \( \eta_{l} \leq \frac{1}{6{\tau}L} \).
% \\
Therefore, we have
% \\
\begin{align}
    \frac{1}{U} \sum_{u=1}^{U} \sum_{k=0}^{\tau-1} \Vert \theta_{u}^{(t,k)} -\theta^{(t)} \Vert^2 \leq 12 \eta_{l}^2 \tau^2 (\tau-1)^2 L (\mathcal{L}(\theta^{(t)}) 
    \notag \\
    - \mathcal{L}(\theta^{*})) + 6 \eta_{l}^2 \tau^2 (\tau-1)^2 \sigma_{*}^2 + \frac{2\tau^3+3\tau^2+\tau}{3} \eta_{l}^2\Vert\vartheta\Vert^2.
\end{align}
\section{Proof of Theorem~\ref{theorem:loss-decrease}} \label{Proof:Theorem3}
We define the following auxiliary variables
\begin{align}
    \text{Aggregate Client Gradient: } r_{u}^{(t)} = \sum_{k=0}^{\tau-1} \nabla \mathcal{L}_{u}(\theta_{u}^{(t,k)})
\end{align}
We also define $R^{(t)} = \frac{1}{U} \sum_{u=1}^{U} r_{u}^{(t)}$. We have
\begin{align}
    \Vert \theta^{(t+1)} - \theta^* \Vert &= \Vert \theta^{(t)} -  \eta_l R^{(t)} - \theta^* \Vert^2
    \\
    &= \Vert \theta^{(t)} - \theta^* \Vert - 2  \eta_l \langle \theta^{(t)} - \theta^*, R^{(t)} \rangle 
    \notag \\
    &+ (\eta_{g}^{(t)})^2 \eta_l^2 \Vert R^{(t)} \Vert^2
    \\
    & \leq \Vert \theta^{(t)} - \theta^* \Vert - 2  \eta_l \underbrace{\langle \theta^{(t)} - \theta^*, R^{(t)} \rangle}_{A1} 
    \notag \\
    &+ \eta_{g}^{(t)} \eta_l^2 \frac{1}{U} \underbrace{\sum_{u=1}^{U} \Vert r_{u}^{(t)} \Vert^2}_{A2}
    \label{A1A2}.
\end{align}

\textbf{Bounding A2: }
We have:
\begin{align}
    A2 &= \sum_{u=1}^{U} \Vert r_{u}^{(t)} \Vert^2
    = \frac{1}{U} \sum_{u=1}^{U} \Vert \sum_{k=0}^{\tau-1} \nabla \mathcal{L}_{u}(\theta_{u}^{(t,k)}) \Vert^2
    \\
    &\leq \frac{\tau}{U} \sum_{u=1}^{U} \sum_{k=0}^{\tau-1} \Vert \nabla \mathcal{L}_{u}(\theta_{u}^{(t,k)}) \Vert^2
    \label{A2-1}
    \\
    &\leq \frac{3 \tau L^2}{U} \sum_{u=1}^{U} \sum_{k=0}^{\tau-1} \Vert \theta_{u}^{t,k} - \theta^{(t)} \Vert + 6\tau^2 L(\mathcal{L}(\theta^{(t)})-\mathcal{L}^*) 
    \notag \\
    &+ 3\tau^2 \sigma_*^2
    \label{A2-2},
\end{align}
where \eqref{A2-1} follows from Jensen inequality and \eqref{A2-2} follows from Lemma 5.

\textbf{Bounding A1: }
\begin{align}
    A1 &= \frac{1}{U} \sum_{u=1}^{U} \langle \theta^{(t)} - \theta^*, r_{u}^{(t)} \rangle\\
    &= \frac{1}{U} \sum_{u=1}^{U} \sum_{k=0}^{\tau-1} \langle \theta^{(t)} - \theta^*, \nabla \mathcal{L}_{u}(\theta_{u}^{(t,k)}) \rangle
    \label{A1}.
\end{align}
We have
\begin{align}
    \langle \theta^{(t)} - \theta^*, \nabla \mathcal{L}_{u}(\theta_{u}^{(t,k)}) \rangle = \langle \theta^{(t)} - \theta_{u}^{(t,k)}, \nabla \mathcal{L}_{u}(\theta_{u}^{(t,k)}) \rangle 
    \notag \\
    + \langle \theta_{u}^{(t,k)} - \theta^*, \nabla \mathcal{L}_{u}(\theta_{u}^{(t,k)}) \rangle.
\end{align}
Because $\mathcal{L}_{u}$ is a $L$-smoothness function, follows \eqref{eq:l-smooth} we have
\begin{align}
    \langle \theta^{(t)} - \theta_{u}^{(t,k)}, \nabla \mathcal{L}_{u}(\theta')\rangle \geq \mathcal{L}_{u}(\theta^{(t)}) -  \mathcal{L}_{u}(\theta^{(t,k)}) 
    \notag \\
    - \frac{L}{2} \Vert \theta^{(t)} - \theta_{u}^{(t,k)} \Vert
    \label{A1-1}.
\end{align}
\text{From convexity $\mathcal{L}_{u}$, we have}
\begin{align}
    \langle \theta^{(t,k)} - \theta^{*}, \nabla \mathcal{L}_{u}(\theta')\rangle \geq \mathcal{L}_{u}(\theta_{u}^{(t,k)}) -  \mathcal{L}_{u}(\theta^{(t)})
    \label{A1-2}.
\end{align}
Therefore, adding \eqref{A1-1} with \eqref{A1-2}, we have 
\begin{align}
    \langle \theta^{(t)} - \theta^*, \nabla \mathcal{L}_{u}(\theta_{u}^{(t,k)}) \rangle \geq \mathcal{L}_{u}(\theta^{(t)}) - \mathcal{L}_{u}(\theta^{*}) 
    \notag \\
    - \frac{L}{2} \Vert \theta^{(t)} - \theta_{u}^{(t,k)} \Vert^2
    \label{A12}
\end{align}
Substituting \eqref{A12} to \eqref{A1} we have
\begin{align}
    A1 \geq \tau(\mathcal{L}(\theta^{(t)}) - \mathcal{L}(\theta^{*})) - \frac{L}{2U} \sum_{u=1}^{U} \sum_{k=0}^{\tau-1} \Vert \theta^{(t)} - \theta_{u}^{(t,k)} \Vert^2.
\end{align}

Substituting the bounds A1 and A2 in \eqref{A1A2} we have 
\begin{align}
    &~~~~\Vert \theta^{(t+1)} - \theta^{*} \Vert^2 \notag \\
    &\leq \Vert \theta^{(t)} - \theta^{*} \Vert^2 - 2  \eta_{l} \tau(1 - 3\eta_lL)(\mathcal{L}(\theta^{(t)} - \mathcal{L}(\theta^{*})) 
    \notag \\
    &+ 3  \eta_l^2 \tau^2 \sigma_*^2 
     + (3  \eta_l^2 \tau L^2) \frac{1}{U} \sum_{u=1}^{U} \sum_{k=0}^{\tau-1}  \Vert \theta^{(t)} - \theta_{u}^{(t,k)} \Vert^2 \label{c1}
    \\
    & \leq \Vert \theta^{(t)} - \theta^{*} \Vert^2 -  \eta_l \tau (\mathcal{L}(\theta^{(t)} - \mathcal{L}(\theta^{*})) +  3  \eta_l^2 \tau^2 \sigma_*^2 
    \label{c2}
    \\ 
    & + 2 \eta_lL \frac{1}{U} \sum_{u=1}^{U} \sum_{k=0}^{\tau-1}  \Vert \theta^{(t)} - \theta_{u}^{(t,k)} \Vert^2 \nonumber
    \\
    & \leq \Vert \theta^{(t)} - \theta^{*} \Vert^2 -  \eta_l \tau (\mathcal{L}(\theta^{(t)} - \mathcal{L}(\theta^{*})) + 3\eta_l^2 \tau^2 \sigma_*^2 
    % + \frac{2\tau^3+3\tau^2+\tau}{3} 2\eta_{l}^3\Vert\vartheta\Vert^2
    \label{c3}
    \\
    & + 24  \eta_{l}^3 \tau^2(\tau-1) L^2 (\mathcal{L}(\theta^{(t)} - \mathcal{L}(\theta^{*})) 
    \notag \\
    &+ 12\eta_{l}^3 \tau^2 (\tau-1) L^2 \sigma_*^2  + \frac{2\tau^3+3\tau^2+\tau}{3} \eta_{l}^3\Vert\vartheta\Vert^2
    \\
    & \leq \Vert \theta^{(t)} - \theta^{*} \Vert^2 - \frac{ \eta_l \tau}{3} (\mathcal{L}(\theta^{(t)} - \mathcal{L}(\theta^{*})) 
    \notag \\
    &+ 3  \eta_l^2 \tau^2 \sigma_*^2 + 12  \eta_{l}^3 \tau^2 (\tau-1) L^2 \sigma_*^2  
    \notag \\
    &+ \frac{2\tau^3+3\tau^2+\tau}{3} \eta_{l}^3\Vert\vartheta\Vert^2.
\end{align}
where \eqref{c1} and \eqref{c3} follows from $\eta_l \leq \frac{1}{6{\tau}L}$ and \eqref{c2} uses Lemma 6. On average over all rounds, we have 
\begin{align}
    \sum_{t=0}^{T-1}  \mathcal{L}(\theta^{(t)}) - \mathcal{L}(\theta^{*}) \leq \frac{3 \Vert \theta^{(0)} - \theta^{*} \Vert^2}{\sum_{t=0}^{T-1}  \eta_l \tau} + 9 \eta_l \tau \sigma_*^2 
    \notag \\
    + 36\eta_l^2\tau(\tau-1)L \sigma_*^2 + (2\tau^2+3\tau+1)\eta_{l}^2\Vert\vartheta\Vert^2.
\end{align}
This implies, 
\begin{align}
    \mathcal{L}(\Bar{\theta^{T})} - \mathcal{L}(\theta^*) \leq \mathcal{O}\Big(\frac{\Vert \theta^{(0)} - \theta^{*} \Vert^2}{\sum_{t=0}^{T-1}  \eta_l \tau}\Big) + \mathcal{O}\Big(\eta_l^2\tau(\tau-1)L \sigma_*^2\Big) 
     \notag \\
    + \mathcal{O}\Big(\eta_l \tau \sigma_*^2\Big) + \mathcal{O}\Big((2\tau^2+3\tau+1)\eta_{l}^2\Vert\vartheta\Vert^2\Big),
\end{align}
where $\Bar{\theta^{T}} = \frac{\sum_{t=0}^{T-1}  \theta^{(t)}}{\sum_{t=0}^{T-1} }$. Apply Theorem~\ref{theorem:batch-unbiased-qubit-grad}, we have: 
\begin{align}
    & \mathcal{L}(\Bar{\theta^{(R)})} - \mathcal{L}(\theta^*) \notag\\
    & \leq \mathcal{O}\Big(\frac{\Vert \theta^{(0)} 
    - \theta^{*} \Vert^2}{\sum_{t=0}^{T-1}  \eta_l \tau}\Big) + \mathcal{O}\Big(\eta_l^2\tau(\tau-1)L \sigma_*^2\Big) \notag \\
    & + \mathcal{O}\Big(\eta_l \tau \sigma_*^2\Big) 
    + \mathcal{O}\Big((2\tau^2+3\tau+1)\eta_{l}^2{3\sigma^2}/{(2^{2n}-1)}\Big).
\end{align}
This completes the proof of Theorem~\ref{theorem:loss-decrease}

\section{Convergence Analysis for Non-Convex Objectives} \label{Proof:Nonconvex}
Our proof technique is inspired by \cite{2020-FL-FedNova} and we use one of their intermediate results to bound the client drift in non-convex settings as we describe below. We highlight the specific steps where we made adjustments to the analysis of \cite{2020-FL-FedNova} below. We begin by defining the following auxiliary variables that will used in the proof.
\begin{align}
    &\text{Quantum Gradient: } h_{u}^{(t)} = \frac{1}{\tau} \sum_{k=0}^{\tau-1} \nabla \ell_{u}(\theta_{u}^{(t,k)}).\\
    &\text{Normalized Gradient: } r_{u}^{(t)} = \frac{1}{\tau} \sum_{k=0}^{\tau-1} \nabla \mathcal{L}_{u}(\theta_{u}^{(t,k)}).
\end{align}
We also define $R^{(t)} = \frac{1}{U} \sum_{u=1}^U r_{u}^{(t)}$

\begin{lemma}[Bounding client drift in Non-Convex Setting]
\begin{align}
    \frac{1}{U} \sum_{u=1}^U \Vert \nabla \mathcal{L}_{u}(\theta^{(t)}) - h_{u}^{(t)} \Vert^2 \leq \frac{1}{8}  \Vert \nabla \mathcal{L}(\theta^{(t)}) \Vert^2 \notag \\
    + 5 \eta_l^2 L^2 \tau(\tau-1) \sigma_g^2.
\end{align}
\end{lemma}
\textit{Proof.} Let $D = 4 \eta_l^2 L^2 \tau(\tau-1)$. From equation (87) in \cite{2020-FL-FedNova}
\begin{align}
    &\frac{1}{U} \sum_{u=1}^U \Vert \nabla \mathcal{L}_{u}(\theta^{(t)}) - h_{u}^{(t)} \Vert^2 \notag\\ 
    &\leq\frac{1}{U} \sum_{u=1}^U \Vert \nabla \mathcal{L}_{u}(\theta^{(t)}) - r_{u}^{(t)} - \sum^{\tau}_{k=1} \vartheta \Vert^2 \notag \\
    &\leq\frac{1}{U} \sum_{u=1}^U \Vert \nabla \mathcal{L}_{u}(\theta^{(t)}) - r_{u}^{(t)} \Vert^2 + \tau \Vert\vartheta\Vert^2  \notag \\
    &\leq \frac{D}{D-1} \Vert \nabla \mathcal{L}_{u}(\theta^{(t)}) \Vert^2 + \frac{D }{1-D}\sigma_g^2 + \tau \Vert\vartheta\Vert^2 \notag \\
    &\leq \frac{D}{D-1} \Vert \nabla \mathcal{L}_{u}(\theta^{(t)}) \Vert^2 + \frac{D }{1-D}\sigma_g^2 + \tau \Vert\vartheta\Vert^2.
\end{align}
From $\eta_l \leq \frac{1}{6{\tau}L}$, we have $D \leq \frac{1}{9}$ so $\frac{1}{1-D} \leq \frac{9}{8}$ and $\frac{D}{1-D} \leq \frac{1}{8}$. Therefore, we have 
\begin{align}
    & \frac{1}{U} \sum_{u=1}^U \Vert \nabla \mathcal{L}_{u}(\theta^{(t)}) - h_{u}^{(t)} \Vert^2 \notag\\ 
    &\leq \frac{1}{8} \Vert \nabla \mathcal{L}_{u}(\theta^{(t)}) \Vert^2 + \frac{9D\sigma_g^2}{8} + \tau \Vert\vartheta\Vert^2\\
    &\leq \frac{1}{8} \Vert \nabla \mathcal{L}_{u}(\theta^{(t)}) \Vert^2 + \frac{9D\sigma_g^2}{8} + \tau \Vert\vartheta\Vert^2\\
    & \leq \frac{1}{8} \Vert \nabla \mathcal{L}_{u}(\theta^{(t)}) \Vert^2 + 5 \eta_l^2 L^2 \tau(\tau-1)\sigma_g^2 + \tau \Vert\vartheta\Vert^2.
\end{align}
Apply Theorem~\ref{theorem:batch-unbiased-qubit-grad}, we have: 
\begin{align}
     \frac{1}{U} \sum_{u=1}^U &\Vert \nabla \mathcal{L}_{u}(\theta^{(t)}) - h_{u}^{(t)} \Vert^2 
    \leq \frac{1}{8} \Vert \nabla \mathcal{L}_{u}(\theta^{(t)}) \Vert^2 \notag\\ 
    &+ 5 \eta_l^2 L^2 \tau(\tau-1)\sigma_g^2 + \tau {3\sigma^2}/{(2^{2n}-1)}.
\end{align}
This completes the proof. 

\section{Proof of Theorem~\ref{theorem:adp-qfl-grad-progress}}\label{appendix:proof-theorem-adp-qfl}
The update of the global model can be written as follows,
\begin{align}
    \theta^{(t+1)} = \theta^{(t)} - \eta_l \tau R^{(t)}.
\end{align}

Now using the $L$-smoothness assumption, we have
\begin{align}
    & \mathcal{L}(\theta^{(t+1)}) - \mathcal{L}(\theta^{(t)}) \notag\\
    &\leq -\eta_l \tau \langle \nabla \mathcal{L}(\theta^{(t)}) , R^{(t)} \rangle + \frac{\eta_l^2 \tau^2 L}{2} \Vert R^{(t)} \Vert^2 
    \\ 
    &\leq -\eta_l \tau \underbrace{\langle \nabla \mathcal{L}(\theta^{(t)}) , R^{(t)} \rangle}_{B1} + \frac{\eta_l^2 \tau^2 L}{2U} \underbrace{\sum_{u=1}^U \Vert h_{u}^{(t)} \Vert^2}_{B2}
    \label{th2-1},
\end{align}
where \eqref{th2-1} followed by $\sum_{u=1}^U \Vert h_{u}^{(t)} \Vert^2$. 

\textbf{Bounding B1:}
\begin{align}
    B1 &= \langle \nabla \mathcal{L}(\theta^{(t)}) , \frac{1}{U} \sum_{u=1}^{U} h_{u}^{(t)} \rangle 
    \\
    &= \frac{1}{2} \Vert \nabla\mathcal{L}(\theta^{(t)}) \Vert^2 + \frac{1}{2} \Vert \frac{1}{U} \sum_{u=1}^{U} h_{u}^{(t)} \Vert^2 
    \notag \\
    &- \frac{1}{2} \Vert \nabla\mathcal{L}(\theta^{(t)}) - \frac{1}{U} \sum_{u=1}^{U} h_{u}^{(t)} \Vert^2
    \\
    & \geq \frac{1}{2} \Vert \nabla\mathcal{L}(\theta^{(t)}) \Vert^2 - \frac{1}{2U} \sum_{u=1}^U \Vert \nabla \mathcal{L}_{u}(\theta^{(t)}) - h_{u}^{(t)} \Vert^2
    \\ 
    & = \frac{1}{2} \Vert \nabla\mathcal{L}(\theta^{(t)}) \Vert^2 - \frac{1}{2U} \sum_{u=1}^U \Vert \nabla \mathcal{L}_{u}(\theta^{(t)}) - h_{u}^{(t)} \Vert^2.
\end{align}

\textbf{Bounding B2:}
we have 
\begin{align}
    B2 &= \sum_{u=1}^U \Vert h_{u}^{(t)} \Vert^2
    \\
    &= \frac{1}{U} \sum_{u=1}^{U} \Vert h_{u}^{(t)} - \nabla \mathcal{L}_{u}(\theta^{(t)}) + \nabla \mathcal{L}_{u}(\theta^{(t)}) 
    \notag \\
    &- \nabla \mathcal{L}(\theta^{(t)}) + \nabla \mathcal{L}(\theta^{(t)}) \Vert^2
    \\
    &\leq \frac{3}{U} \sum_{u=1}^{U} ( \Vert h_{u}^{(t)} - \nabla \mathcal{L}_{u}(\theta^{(t)}) \Vert^2 + \Vert \nabla \mathcal{L}_{u}(\theta^{(t)}) 
    \notag \\
    &- \nabla \mathcal{L}(\theta^{(t)}) \Vert^2 + \Vert \nabla \mathcal{L}(\theta^{(t)}) \Vert^2 )
    \label{b2-1}
    \\
    % & \leq \sum_{u=1}^{U} \Vert h_{u}^{(t)} - \nabla \mathcal{L}_{u}(\theta^{(t)}) \Vert^2 + 3 \sigma_g + 3\Vert \nabla \mathcal{L}(\theta^{(t)}) \Vert^2 \notag \\
    % & \leq \sum_{u=1}^{U} \Vert r_{u}^{(t)} + \sum^{\tau}_{k=1}\vartheta - \nabla \mathcal{L}_{u}(\theta^{(t)}) \Vert^2 + 3 \sigma_g + 3\Vert \nabla \mathcal{L}(\theta^{(t)}) \Vert^2 
    % \label{b2-2} \\
    & \leq \frac{3}{U} \sum_{u=1}^{U} \Vert h_{u}^{(t)} - \nabla \mathcal{L}_{u}(\theta^{(t)}) \Vert^2 + 3 \sigma_g + 3\Vert \nabla \mathcal{L}(\theta^{(t)}) \Vert^2, \label{b2-2}
\end{align}
where \eqref{b2-1} uses Jensen's inequality, \eqref{b2-2} uses bounded heterogeneity assumption. We would like to note that the bound for B2 is our contribution and is needed in our proof due to the relaxation in \eqref{A1A2}. The bound for B1 follows a similar technique as in \cite{2020-FL-FedNova}.

Substituting the B1 and B2 bounds into \eqref{th2-1}, we have
\begin{align}
    & \mathcal{L}(\theta^{(t+1)}) - \mathcal{L}(\theta^{(t)}) \notag\\
    &\leq -  \eta_l \tau \frac{1}{2} \Vert\nabla\mathcal{L}(\theta^{(t)}) \Vert^2 + \frac{1}{2U} \Vert \nabla \mathcal{L}_{u}(\theta^{(t)}) - h_{u}^{(t)} \Vert^2 \notag\\
    & \quad+ \frac{\eta_l \tau L}{2} \Big(3 \sigma_g^2 + 3 \Vert \nabla\mathcal{L}(\theta^{(t)}) \Vert^2 
    \notag \\
    & \quad+ \frac{3}{U} \sum_{u=1}^{U} ( \Vert h_{u}^{(t)} - \nabla \mathcal{L}_{u}(\theta^{(t)}) \Vert^2)\Big) \nonumber
    \\
    & \leq  -  \eta_l \tau \Big(\frac{1}{4} \Vert \nabla\mathcal{L}(\theta^{(t)}) \Vert^2 
    \notag \\
    & \quad+ \frac{3}{U} \Vert \nabla \mathcal{L}_{u}(\theta^{(t)}) - h_{u}^{(t)} \Vert^2 + 3 \eta_l \tau L \sigma_g^2 \Big) 
    \\
    & \leq  -  \eta_l \tau \Big( \frac{1}{4} \Vert \nabla\mathcal{L}(\theta^{(t)}) \Vert^2 
    \notag \\
    & \quad+ \frac{3}{U} \Big\Vert \nabla \mathcal{L}_{u}(\theta^{(t)}) - r_{u}^{(t)} + \sum^{\tau}_{k=1}\vartheta\Big\Vert^2 + 3 \eta_l \tau L \sigma_g^2 \Big) 
    \label{th2-a}
    \\
    & \leq  -  \eta_l \tau \Big( \frac{1}{4} \Vert \nabla\mathcal{L}(\theta^{(t)}) \Vert^2 
    \notag \\
    & \quad+ \frac{3}{U} \Big\Vert \nabla \mathcal{L}_{u}(\theta^{(t)}) - r_{u}^{(t)}\Big\Vert^2 + 3 \eta_l \tau L \sigma_g^2  + \Vert\sum^{\tau}_{k=1}\vartheta\Vert^2\Big) 
    \\
    & \leq  -  \eta_l \tau \Big( \frac{1}{4} \Vert \nabla\mathcal{L}(\theta^{(t)}) \Vert^2 \notag \\
    & \quad+ \frac{3}{U} \Big\Vert \nabla \mathcal{L}_{u}(\theta^{(t)}) - r_{u}^{(t)}\Big\Vert^2 + 3 \eta_l \tau L \sigma_g^2  + \sum^{\tau}_{k=1}\Vert\vartheta\Vert^2\Big)     
    \\    
    & \leq -  \eta_l \tau \Big( \frac{1}{8} \Vert \nabla\mathcal{L}(\theta^{(t)}) \Vert^2 + 3 \eta_l \tau L \sigma_g^2 
    \notag \\
    & \quad+ 5\eta_l^2 \tau(\tau-1) L^2 \sigma_g^2 + \sum^{\tau}_{k=1}\tau \Vert\vartheta\Vert^2\Big),
    \label{th2-b}
\end{align}
where \eqref{th2-a} uses $\eta_l \leq \frac{1}{6{\tau}L}$ and \eqref{th2-b} uses Lemma 7. Thus arranging terms and averaging over all rounds, we have
\begin{align}
    \mathcal{L}(\theta^{(0)}) - \mathcal{L}(\theta^{(t)}) 
    & \leq \eta_l \tau \Big( \frac{1}{8} \sum^{T}_{t=0}\Vert\nabla\mathcal{L}(\theta^{(t)}) \Vert^2 + 3 \sum^{T}_{t=0}\eta_l \tau L \sigma_g^2 
    \notag \\
    &+ 5\sum^{T}_{t=0}\eta_l^2 \tau(\tau-1) L^2 \sigma_g^2 + \sum^{T}_{t=0}\sum^{\tau}_{k=1}\Vert\vartheta\Vert^2\Big) \notag\\
    & = \eta_l \tau \Big( \frac{1}{8} \sum^{T}_{t=0}\Vert\nabla\mathcal{L}(\theta^{(t)}) \Vert^2 
    + T\tau^2\Vert\vartheta\Vert^2
    \notag \\
    &+ 3 T\eta_l \tau L \sigma_g^2 + 5T\eta_l^2 \tau(\tau-1) L^2 \sigma_g^2 \Big) 
\label{eq:non-convex-convergence}.
\end{align}
To characterize the gradient progress when the circumstance of being trapped into the local minimizers, we propose a probability $P_\textrm{sharp}$ which is the probability of being trapped into the sharp minimizers. This sharp minimizer is being capped by the $L$-smooth. Thus, the condition is that, when the batch-wise gradient surpasses the $L$ value, the model can continue to progress. Otherwise, the gradient is trapped for the latter phase. Thus, ~\eqref{eq:non-convex-convergence} can be reformulated as:
\begin{align}
    &~~~~\mathcal{L}(\theta^{(0)}) - \mathcal{L}(\theta^{(t)}) 
    \leq \eta_l \tau \Big( \frac{1}{8} \Big[\sum^{T}_{t=0}\Vert\nabla\mathcal{L}(\theta^{(t)})\Vert^2 
    \notag \\
    &-\sum^{T}_{t=m}P_\textrm{sharp}\mathbbm{1}(\nabla\mathcal{L}(\theta^{(m)})+\vartheta < L)\Vert\nabla\mathcal{L}(\theta^{(t)}) \Vert^2\Big] \notag \\
    &+ T\tau^2\Vert\vartheta\Vert^2
    + 3 T\eta_l \tau L \sigma_g^2 + 5T\eta_l^2 \tau(\tau-1) L^2 \sigma_g^2 \Big).
\end{align}
Thus rearranging terms and averaging over all rounds we have 
\begin{align}
    &\frac{\sum^{T-1}_{t=0} \Vert\nabla\mathcal{L}(\theta^{(t)}) \Vert^2}{T}\notag\\
    &\leq 
      \frac{8(\mathcal{L}(\theta^{(0)}) - \mathcal{L}(\theta^{(t)}))}{T\eta_l \tau 
      (1-\sum^{T}_{t=m}P_\textrm{sharp}\mathbbm{1}(\nabla\mathcal{L}(\theta^{(m)})+\vartheta < L))} \notag \\
    &+ \frac{24 T\eta_l \tau L \sigma_g^2 
     + 40T\eta_l^2 \tau(\tau-1) L^2 \sigma_g^2}{(1-\sum^{T}_{t=m}P_\textrm{sharp}\mathbbm{1}(\nabla\mathcal{L}(\theta^{(m)})+\vartheta < L))}  .
\end{align}
Let $\kappa = (1-\sum^{T}_{t=m}P_\textrm{sharp}\mathbbm{1}(\nabla\mathcal{L}(\theta^{(m)})+\vartheta < L))$ . This implies,
\begin{align}
    &\min_{t\in[T]}\Vert\nabla\mathcal{L}(\theta^{(t)}) \Vert^2 
    \leq 
    \mathcal{O}\Big(\frac{8(\mathcal{L}(\theta^{(0)}) - \mathcal{L}(\theta^{(t)}))}{T\eta_l\tau\kappa}\Big) \notag \\
    &+ 
    \mathcal{O}(\frac{\eta_l^2 \tau(\tau-1) L^2 \sigma_g^2}{\kappa}) +
    \mathcal{O}(\frac{24 T\eta_l \tau L \sigma_g^2}{\kappa}) .
\end{align}
This completes the proof.

\end{document}